\documentclass[11pt,a4paper]{article}
\usepackage[utf8]{inputenc}
\usepackage[linesnumbered,ruled,vlined]{algorithm2e}
\usepackage{amsmath,amssymb,amsthm,mathtools}
\usepackage{booktabs}
\usepackage{makecell}
\usepackage{threeparttable}

\usepackage[noadjust]{cite}
\usepackage{array,multirow,graphicx}
\usepackage{comment}
\usepackage[svgnames]{xcolor}
\usepackage{tikz,pgfplots,pgfplotstable}
\usetikzlibrary{positioning,intersections,calc,arrows,decorations.markings,patterns}

\newcommand{\ot}{\leftarrow}
\newcommand{\argmin}{\mathop{\rm arg\,min}}
\newcommand{\argmax}{\mathop{\rm arg\,max}}
\newcommand{\OPT}{\mathrm{OPT}}
\newcommand{\ALG}{\mathrm{ALG}}
\pgfplotsset{compat=newest}

\theoremstyle{definition}
\newtheorem{theorem}{Theorem}[section]
\newtheorem{lemma}{Lemma}[section]

\newtheorem{corollary}{Corollary}[section]

\newtheorem{claim}{Claim}[section]
\newtheorem{observation}{Observation}[section]
\usepackage[bookmarks=false]{hyperref}
\hypersetup{
    colorlinks=true,
    linkcolor=blue,
    citecolor=blue,
}
\usepackage{fullpage}

\title{Online Knapsack Problems with a Resource Buffer}
\author{Xin Han\thanks{Dalian University of Technology, China, Email: {\tt hanxin@dlut.edu.cn}} \and
  Yasushi Kawase\thanks{Tokyo Institute of Technology, Japan, Email: {\tt kawase.y.ab@m.titech.ac.jp}} \and
  Kazuhisa Makino\thanks{Kyoto University, Japan, Email: {\tt makino@kurims.kyoto-u.ac.jp}} \and
  Haruki Yokomaku\thanks{NTT DATA Mathematical Systems, Japan, Email: {\tt dsm4up2c@gmail.com}}
}
\date{}

\begin{document}
\maketitle
\begin{abstract}
In this paper, we introduce online knapsack problems with a resource buffer.
In the problems, we are given a knapsack with capacity $1$, a buffer with capacity $R\ge 1$, and items that arrive one by one.
Each arriving item has to be taken into the buffer or discarded on its arrival irrevocably.
When every item has arrived, we transfer a subset of items in the current buffer into the knapsack.
Our goal is to maximize the total value of the items in the knapsack.
We consider four variants depending on whether items in the buffer are
removable (i.e., we can remove items in the buffer) or non-removable, and proportional (i.e., the value of each item is proportional to its size) or general.
For the general\&non-removable case, we observe that no constant competitive algorithm exists for any $R\ge 1$.
For the proportional\&non-removable case, we show that a simple greedy algorithm is optimal for every $R\ge 1$.
For the general\&removable and the proportional\&removable cases, we present optimal algorithms for small $R$ and give asymptotically nearly optimal algorithms for general $R$.
\end{abstract}

\section{Introduction}
Online knapsack problem is one of the most fundamental problems in online optimization~\cite{KelPfePis04,Komm2016}.
In the problem, we are given a knapsack with a fixed capacity, and items with sizes and values, which arrive one by one.
Upon arrival, we must decide whether to accept the arrived item into the knapsack, and this decision is irrevocable.

In this paper, we introduce a variant of the online knapsack problem, which we call \emph{online knapsack problems with a resource buffer}.
Suppose that we have a buffer with fixed capacity in addition to a knapsack with fixed capacity, and items arrive online.
Throughout this paper, we assume that the knapsack capacity is $1$, and the buffer capacity is $R~(\ge 1)$.
In addition, assume that each item $e$ has a size $s(e)$ and a value $v(e)$.
When an item $e$ has arrived, we must decide whether to take it into the buffer or not. 
The total size of the selected items must not exceed the capacity of the buffer $R$.
Further, we cannot change the decisions that we made past, i.e., once an item is rejected, it will never be put into the buffer.
We consider two settings:
(i) \emph{non-removable}, i.e., we cannot discard items in the buffer, and
(ii) \emph{removable}, i.e., we can discard some items in the buffer, and once an item is discarded, it will never be put into the buffer again.
After the end of the item sequence, we transfer a subset of items from the buffer into the knapsack.
Our goal is to maximize the total value of the items in the knapsack under the capacity constraint.
It is worth mentioning that, if $R=1$, our problem is equivalent to the standard online knapsack problem.

Our model can be regarded as a ``partial'' resource augmentation model. 
That is, in the resource augmentation model, the online algorithm can use the buffer for the final result.
On the other hand, in our model, the online algorithm uses the buffer only to temporary store items, and it must use the knapsack to output the final result.
Moreover, our model can be viewed as a streaming setting: we process items in a streaming fashion, and we can keep only a small portion of the items in memory at any point.

To make things more clear, let us see an example of the online knapsack problem with a resource buffer.
Let $R=1.5$.
Suppose that three items $e_1,e_2,e_3$ with
$(s(e_1),v(e_1))=(0.9,4)$, $(s(e_2),v(e_2))=(0.7,3)$, $(s(e_3),v(e_3))=(0.2,2)$
are given in this order,
but we do not know the items in advance.
When $e_1$ has arrived, suppose that we take it into the buffer.
Then, for the non-removable case, we need to reject $e_2$ because we cannot put it together with $e_1$.
In contrast, for the removable case, we have another option---take $e_2$ into the buffer by removing $e_1$.
If $\{e_1,e_3\}$ is selected in the buffer at the end, the resulting value is $4$ by transferring $\{e_1\}$ to the knapsack.
Note that, in the resource augmentation model, we can obtain a solution with value $6$ by selecting $\{e_1,e_3\}$.

\subsection*{Related work}
For the non-removable online knapsack problem (i.e., non-removable case with $R=1$),
Marchetti-Spaccamela and Vercellis~\cite{MV1995} showed that no constant competitive algorithm exists.
Iwama and Taketomi~\cite{IT2002} showed that there is no constant competitive algorithm even for the proportional case (i.e., the value of each item is proportional to its size).
The problem has also studied under some restrictions on the input~\cite{KP1998,BIKK2007,ZCL2008,AKL2019}.

The removable variant of the online knapsack problem (i.e., removable case with $R=1$) is introduced by Iwama and Taketomi~\cite{IT2002}.
They proved that no constant competitive deterministic algorithm exists in general,
but presented an optimal $(1+\sqrt 5)/2$-competitive algorithm for the proportional case.
The competitive ratios can be improved by using randomization~\cite{HKM2015,CJS2016}.
In addition, the problem with removal cost has been studied under the name of the \emph{buyback problem}~\cite{BHK2009,AK2009,HKM2014,KHM2016,KHM2018}.

An online knapsack problem with resource augmentation is studied by Iwama and Zhang~\cite{IZ2007}.
In their setting, an online algorithm is allowed to use a knapsack with capacity $R\ge 1$, while the offline algorithm has a knapsack with capacity $1$.
They developed optimal $\max\{1,\,1/(R-1)\}$-competitive algorithms for the general\&removable and proportional\&non-removable cases and
an optimal $\max\bigl\{1,\,\min\{\frac{1+\sqrt{4R+1}}{2R},\frac{2}{2R-1}\}\bigr\}$-competitive algorithm for the proportional\&removable case.
All of their algorithms are based on simple greedy strategies.
The competitive ratios except for the general\&non-removable cases become exactly $1$ when $R$ is a sufficiently large real.

In addition, 
there exist several papers that apply online algorithms to approximately solve the constrained stable matching problems~\cite{KI2017,KI2018,KI2019}.

\subsection*{Our results}
We consider four variants depending on whether removable or non-removable, and proportional or general.
In this paper, we focus on deterministic algorithms.
Our results are summarized in Table~\ref{table:summary}.
To compare our model to the resource augmentation model, we list the competitive ratio for both models in the table.
It should be noted that each competitive ratio in our model is at least the corresponding one in the resource augmentation model.
Hence, lower bounds for the resource augmentation model are also valid to our model.

For the general\&non-removable case, we show that there is no constant competitive algorithm.
%
For the proportional\&non-removable case, we show that a simple greedy is optimal and its competitive ratio is $\max\{2,1/(R-1)\}$.
Interestingly, the competitive ratio is equal to the ratio in resource augmentation model for $1<R\le 3/2$.
%
For the general\&removable case, we present an optimal algorithm for $1<R\le 2$. 
Furthermore, for large $R$, we provide an algorithm that is optimal up to a logarithmic factor. 
The algorithm partitions the input items into groups according to sizes and values, and it applies a greedy strategy for each group that meets a dynamically adjusted threshold.
We will see that the competitive ratio is larger than $1$ for any $R$ but it converges to $1$ as $R$ goes to infinity.
%
For the proportional\&removable case, we develop optimal algorithms for $1\le R\le 3/2$.
The basic idea of the algorithms is similar to that of the algorithm for $R=1$ given by Iwama and Taketomi~\cite{IT2002}.
Our algorithms classify the items into three types---\emph{small}, \emph{medium}, and \emph{large}---and the algorithms carefully treat medium items.
We observe that, as $R$ becomes large, we need to handle more patterns to obtain an optimal algorithm.
In addition, for large $R$, we show that the algorithm for the general\&removable case is also optimal up to a logarithmic factor.

\begingroup
\newcommand{\thm}[1]{{\footnotesize(Thm.~\ref{#1})}}
\newcommand{\thms}[2]{{\footnotesize(Thms.~\ref{#1},\,\ref{#2})}}
\newcommand{\cor}[1]{{\footnotesize(Cor.~\ref{#1})}}
\renewcommand{\arraystretch}{1.2}
\setlength{\tabcolsep}{5pt}
\begin{table}[ht]
\centering
\caption{Summary of the competitive ratios for our model and the resource augmentation model.}
\label{table:summary}
\begin{threeparttable}
	\scalebox{0.77}{
	\begin{tabular}{c|c||ccc|ccc}\toprule
		\multicolumn{2}{c||}{\multirow{2}{*}{variants}} & \multicolumn{3}{c|}{Our model}& \multicolumn{3}{c}{Resource augmentation}\\\cline{3-8}
		\multicolumn{2}{c||}{}   & $R$ & \makecell{lower\\bound} & \makecell{upper\\bound} & $R$ & \makecell{lower\\bound} & \makecell{upper\\bound} \\ \midrule
		\multirow{4}{*}{\rotatebox[origin=c]{90}{\small non-removable}}& \multirow{3}{*}{prop.}& $1$ & $\infty$~\cite{MV1995}&--- & $1$ & $\infty$~\cite{MV1995} & --- \\
		&  & $(1,\,\frac{3}{2}]$ & $\frac{1}{R-1}$~\cite{IZ2007} & $\frac{1}{R-1}$~\thm{thm:non_removable} & $(1,\, 2]$ & $\frac{1}{R-1}$~\cite{IZ2007}&$\frac{1}{R-1}$~\cite{IZ2007} \\

		& & $[\frac{3}{2},\,\infty)$ & 2~\thm{thm:lower_non_ge} & 2~\cor{cor:non_removable} & $[2,\, \infty)$ & 1 & 1~\cite{IZ2007}\\ \cline{2-8}

		& gen. &  $[1,\,\infty)$ &$\infty$~\cite{MV1995} & --- & $[1,\, \infty)$ &$\infty$~\cite{MV1995} & --- \\\hline

		\multirow{11}{*}{\rotatebox[origin=c]{90}{\small removable}}& \multirow{7}{*}{prop.} &$1$ & $\frac{1+\sqrt{5}}{2}$~\cite{IT2002}& $\frac{1+\sqrt{5} }{2}$~\cite{IT2002} &$1$ & $\frac{1+\sqrt{5} }{2}$~\cite{IT2002}& $\frac{1+\sqrt{5} }{2}$~\cite{IT2002}\\

		& & $[1,\,\frac{1+\sqrt{2}}{2}]$ & $\frac{1+\sqrt{4R+1}}{2R}$~\thm{thm:lower_1leRle2} & $\frac{1+\sqrt{4R+1}}{2R} $~\thms{thm:proportional_removable_<10/9}{thm:10/9_to_(1+sqrt 2)/2} &$[1,\, \frac{1+\sqrt{2}}{2}]$ & $\frac{1+\sqrt{4R+1}}{2R}$~\cite{IZ2007} & $\frac{1+\sqrt{4R+1}}{2R}$~\cite{IZ2007}\\

		& & $[\frac{1+\sqrt{2}}{2},\, 2- \frac{\sqrt{2}}{2}]$ & $\sqrt 2$~\thm{thm:prop/removable_lower_ii} & $\sqrt 2$~\thm{thm:10/9_to_(1+sqrt 2)/2} & \multirow{4}{*}{$[\frac{1+\sqrt{2}}{2},\, \frac{3}{2}]$} & \multirow{4}{*}{$\frac{2}{2R-1}$~\cite{IZ2007}} & \multirow{4}{*}{$\frac{2}{2R-1}$~\cite{IZ2007}}\\

		& & $[2- \frac{\sqrt{2}}{2},\, 17-9\sqrt{3}]$ & $\frac{\sqrt{16R+1}-1}{2R}$~\thm{thm:prop/removable_lower_ii} & $\frac{\sqrt{16R+1}-1}{2R}$~\thm{thm:17-9sqrt(3)} & & & \\

		& & $[17-9\sqrt{3},\, 2\sqrt{3}-2]$ & $\frac{1+\sqrt{3}}{2}$~\thm{thm:prop/removable_lower_iii} & $\frac{1+\sqrt{3}}{2}$~\thm{thm:17-9sqrt(3)} & & & \\

		& & $[2\sqrt{3}-2,\, \frac{3}{2}]$ & $\frac{2}{R}$~\thm{thm:prop/removable_lower_iii} & $\frac{2}{R}$~\thm{thm:2/R} & & & \\

		& & $[1,\, \infty)$ & $1+\frac{1}{\lceil 2R\rceil +1}$~\thm{thm:lower_proportional_removable_general} & $1+O(\frac{\log R}{R})$~\thm{thm:general_removable_general} & $[\frac{3}{2},\,\infty)$ & $1$ & $1$~\cite{IZ2007} \\\cline{2-8}

		&\multirow{4}{*}{gen.} &  $1$& $\infty$~\cite{MV1995} & --- & $1$ &$\infty$~\cite{MV1995} & ---  \\

		&  & $(1,\, \frac{3}{2}]$ & $\frac{1}{R-1}$~\thm{thm:general/removable_lower_1<R<3/2} & $\frac{1}{R-1}$~\thm{thm:weighted_removable_to_2} &\multirow{2}{*}{$(1,\,2]$} & \multirow{2}{*}{$\frac{1}{R-1}$~\cite{IZ2007}} & \multirow{2}{*}{$\frac{1}{R-1}$~\cite{IZ2007}}\\

		& & $[\frac{3}{2},\, 2)$ & 2~\thm{thm:general/removable_lower_3/2<R<2} & 2~\thm{thm:weighted_removable_to_2} & & & \\

		& & $[1,\, \infty)$ & $1+\frac{1}{R+1}$~\thm{thm:general/removable_lower_R>1} & $1+O(\frac{\log R}{R})$~\thm{thm:general_removable_general} &$[2,\,\infty)$ & 1 & 1~\cite{IZ2007}\\\bottomrule
	\end{tabular}
                    }
\end{threeparttable}
\end{table}
\endgroup

\section{Preliminaries}\label{section:preliminaries}
We denote the size and the value of an item $e$ as $s(e)$ and $v(e)$, respectively.
We assume that $1\ge s(e)>0$ and $v(e)\ge 0$ for any $e$.
For a set of items $B$, we abuse notation, and let $s(B)=\sum_{e\in B}s(e)$ and $v(B)=\sum_{e\in B}v(e)$.

For an item $e$, the ratio $v(e)/s(e)$ is called the \emph{density} of $e$.
If all the given items have the same density, we call the problem \emph{proportional}.
Without loss of generality, we assume that $v(e) = s(e)$ for the proportional case.
We sometimes represent an item $e$ as the pair of its size and value $(s(e),v(e))$.
Also, for the proportional case, we sometimes represent an item $e$ as its size $s(e)$.

Let $I=(e_1,\dots,e_n)$ be the input sequence of the online knapsack problem with a resource buffer.
For a deterministic online algorithm $\ALG$, let $B_i$ be the set of items in the buffer at the end of the round $i$.
Note that $B_0=\emptyset$.
In the removable setting, they must satisfy $B_i\subseteq B_{i-1}\cup\{e_i\}$ and $s(B_i)\le R$ $(i=1,\dots,n)$.
In the non-removable setting, they additionally satisfy $B_{i-1}\subseteq B_i$ $(i=1,\dots,n)$.
Without loss of generality, we assume that the algorithm transfers the optimal subset of items from the buffer into the knapsack since we do not require the online algorithm to run in polynomial time.
We denote the outcome value of $\ALG$ by $\ALG(I)~(\coloneqq \max\{v(B)\mid B\subseteq B_n,~s(B)\le 1\})$
and the offline optimal value $\OPT(I)~(\coloneqq \max\{v(B)\mid B\subseteq \{e_1,\dots,e_n\},~s(B)\le 1\})$.
Then, the \emph{competitive ratio} of $\ALG$ for $I$ is defined as $\OPT(I)/\ALG(I)~(\ge 1)$.
In addition, the competitive ratio of a problem is defined as $\inf_{\ALG}\sup_{I}\OPT(I)/\ALG(I)$,
where the infimum is taken over all (deterministic) online algorithms and the supremum is taken over all input sequences.

\section{General\&Non-removable Case}\label{section:general/non-removable}
To make the paper self-contained, we show that the general\&non-removable case admits no constant competitive algorithm.
To see this, we observe an input sequence given by Iwama and Zhang~\cite{IZ2007},
which was used to prove the corresponding result for the resource augmentation setting.

\begin{theorem}
  For any $R\ge 1$,
  there exists no constant competitive algorithm for the general\&non-removable online knapsack problem with a buffer.
\end{theorem}
\begin{proof}
Let $ALG$ be an online algorithm and let $R\ge 1$ and $c$ be positive reals. 
Consider the input sequence \(I\coloneqq ((1,c^1),\,(1,c^2),\dots,\,(1,c^k))\),
where $(1,c^k)$ is the first item so that $\ALG$ does not take into the buffer.
Note that $k\le \lfloor R\rfloor + 1$ since the buffer size is $R$.
If $k=1$, $\ALG$ is not competitive, since $\ALG(I) = 0$ and $\OPT(I) = c$.
If $k>1$, since $\ALG(I) = c^{k-1}$ and $\OPT(I) = c^k$, the competitive ratio is $c$, which is unbounded as $c$ goes to infinity.
\end{proof}

\section{Proportional\&Non-removable Case}\label{section:prop/non-removable}
In this section, we consider the proportional\&non-removable case.
We show that the competitive ratio is $\max\{\frac{1}{R-1},\,2\}$ for the case.

\subsection{Lower bounds}
For lower bounds, we consider two cases separately: $1<R\le 3/2$ and $R>3/2$.

\begin{theorem}\label{lower_bound_non_removable_Rle3_2}
For all $R$ with $1<R \le 3/2$ and all $\epsilon > 0$, 
the competitive ratio of the proportional\&non-removable online knapsack problem with a buffer is 
at least $1/(R-1)-\epsilon$.
\end{theorem}
\begin{proof}
Let $\epsilon'$ be a positive real such that $\frac{1}{R-1+\epsilon'}\ge \frac{1}{R-1}-\epsilon$ and 
let $\ALG$ be an online algorithm.
Consider the following input sequence $I$: \[R-1 + \epsilon',\ 1.\]
Then, $\ALG$ must pick the first item, otherwise $\ALG$ is not competitive, since $\ALG(I) = 0$ and $\OPT(I) = R-1+\epsilon'$.
Recall that $\ALG$ cannot discard the item since we consider the non-removable setting.
Also, $\ALG$ cannot take the second item since the buffer size is strictly smaller than the total size of the first and the second items.
Thus, $\ALG(I) = R-1+\epsilon'$ and $\OPT(I) = 1$, and hence the competitive ratio is at least $\frac{1}{R-1 + \epsilon'} \ge \frac{1}{R-1}-\epsilon.$
\end{proof}
It should be noted that the input sequence in the proof of Theorem~\ref{lower_bound_non_removable_Rle3_2} 
is the same as the one in \cite{IZ2007}, which is used to show a lower bound for the resource augmentation model. 

\begin{theorem}\label{thm:lower_non_ge}
For all $R> 3/2$ and all $\epsilon > 0$, 
the competitive ratio of the proportional\&non-removable online knapsack problem with a buffer is 	
at least $2-\epsilon$.
\end{theorem}
\begin{proof}
Let $\epsilon'$ be a positive real such that $\frac{2}{1+2\epsilon'}\ge 2-\epsilon$ and 
let $\ALG$ be an online algorithm.
Consider the following input sequence $I$:
\[\frac{1}{2} + \epsilon',\ \frac{1}{2} + \frac{\epsilon'}{2},\,\ldots\,,\ \frac{1}{2} + \frac{\epsilon'}{k},\ \frac{1}{2}-\frac{\epsilon'}{k},\]
where the $k$th item ($1/2+\epsilon'/k$) is the first item that $\ALG$ does not take it into the buffer.
Note that $I$ is uniquely determined by $\ALG$ and $k\le 2R$.
Since $\ALG(I) = 1/2+\epsilon'$ and $\OPT(I) = 1/2+\epsilon'/k + 1/2- \epsilon'/k = 1$, the competitive ratio is at least $\frac{1}{1/2+\epsilon'}\ge 2 - \epsilon$.
\end{proof}

\subsection{Upper bounds}
For upper bounds, we consider an algorithm that greedily picks a given item if it is possible.
The formal description of the algorithm is given in Algorithm~\ref{alg:non_removable}.
Recall that the resulting outcome of the algorithm is $\max\{s(B)\mid B\subseteq B_n,\,s(B)\le 1\}$, where $B_n$ is the items in the buffer at the final round $n$.
We prove that the algorithm is optimal for any $R>1$.
\begin{algorithm}[htb]
  \caption{$1/(R-1)$-competitive algorithm}\label{alg:non_removable} 
  $B_0\ot\emptyset$\;
  \For{\(i\ot 1,2,\dots\)}{%
    \lIf{$s(B_{i-1}\cup\{e_i\})\le R$}{$B_i\ot B_{i-1}\cup\{e_i\}$ \textbf{else} $B_i\ot B_{i-1}$}
  }
\end{algorithm}

\begin{theorem}\label{thm:non_removable}
  Algorithm~\ref{alg:non_removable} is $1/(R-1)$-competitive 
  for the proportional\&non-removable online knapsack problem with a buffer 
  when $1<R\le 3/2$.
\end{theorem}
\begin{proof}
Let $\ALG$ be an online algorithm induced by Algorithm~\ref{alg:non_removable} and $I$ be an input sequence.
Without loss of generality, we can assume $s(I)>R$ since otherwise $\ALG(I)=\OPT(I)$.

Suppose that $I$ does not contain items with size at least $R-1$.
Let $k$ be the round such that $\sum_{i=1}^{k-1}s(e_i) < R-1 \le \sum_{i=1}^k s(e_i)$.
Then, we have $s(B_{k})=\sum_{i=1}^{k}s(e_i)=s(e_k)+\sum_{i=1}^{k-1}s(e_i)< (R-1)+(R-1) \le 1$ by $s(e_k)<R-1$ and $R\le 3/2$.
Therefore, in this case, the competitive ratio is at most $\frac{1}{R-1}$.

Next, suppose that $I$ contains an item with size at least $R-1$.
Let $e_j$ be the first item in $I$ such that $s(e_j) \ge R-1$.
If $s(B_{j-1})\ge R-1$, then the competitive ratio is at most $\frac{1}{R-1}$ by the same argument as above.
Otherwise (i.e., $s(B_{j-1})< R-1$), we have $s(B_{j-1}\cup\{e_j\})\le R$ and hence $e_j\in B_j\subseteq B_n$, i.e., $e_j$ is selected in $B_n$.

Thus, $\ALG(I)\ge s(e_j) = R-1$ and the competitive ratio is at most $\frac{1}{R-1}$.
\end{proof}

Since $1/(R-1) = 2$ when $R = 3/2$,
we obtain the following corollary from Theorem~\ref{thm:non_removable}.
\begin{corollary}\label{cor:non_removable}
  Algorithm~\ref{alg:non_removable} is $2$-competitive 
  for the proportional\&non-removable online knapsack problem with a buffer 
  when $R\ge 3/2$.
\end{corollary}

\section{General\&Removable Case}\label{section:weighted removable case}
In this section, we consider the general\&removable case.
We show that the competitive ratio is $\max\{\frac{1}{R-1},\,2\}$ for $R\le 2$.
In addition, for general $R$, we prove that the competitive ratio is
at most $1+O(\log R/R)$ and at least $1+\frac{1}{R+1}$.

\subsection{Lower bounds}
Here, we give lower bounds of the competitive ratio in this case.
We first present a general lower bound $1+1/(R+1)$.
\begin{theorem}\label{thm:general/removable_lower_R>1}
For $R\ge 1$, the competitive ratio of the general\&removable online knapsack problem with a buffer is at least $1+\frac{1}{R+1}$.
\end{theorem}
\begin{proof}
Let $\ALG$ be an online algorithm.
Let $\epsilon$ be a real such that $0<\epsilon <\min\{\frac{2}{R(R+1)},\, \lfloor R\rfloor+1-R\}~(\le 1)$ and $n = \lceil R+\epsilon \rceil~(\le R+1)$.
Consider the following item sequence $I$:
\[(1, 1),\ \Bigl(1- \epsilon^2, 1- \tfrac{1}{n}\Bigr),\,\ldots\,,\ \Bigl(1-i\epsilon^2, 1- \tfrac{i}{n}\Bigr),\,\dots\,,\ \Bigl(1- (n-1)\epsilon^2, 1- \tfrac{n-1}{n}\Bigr).\]
Since the total size of the items in this sequence is $n-\frac{(n-1)n}{2}\epsilon^2 > R+\epsilon - \frac{R(R+1)}{2} \epsilon^2 > R$, $\ALG$ cannot store all the items and must discard at least one of them.
Let $\left(1- i\epsilon^2, 1- \frac{i}{n}\right) \ (0\le i < n)$ be a discarded item
and consider an item sequence $I'$ in which $\left(i\epsilon^2, \frac{i+1}{n}\right)$ is given after $I$. 
Then, $\OPT(I') = (1-\frac{i}{n})+(\frac{i+1}{n})=1+\frac{1}{n}$ and $\ALG(I') \le 1$.
Therefore, the competitive ratio is at least $1+\frac{1}{n} = 1+\frac{1}{\lceil R+\epsilon \rceil} \ge 1+\frac{1}{R+1}$.
\end{proof}

Next, we provide the tight lower bound for $R\le 2$.
We separately consider the following two cases: $1<R\le 3/2$ and $3/2\le R< 2$.
\begin{theorem}\label{thm:general/removable_lower_1<R<3/2}
For all $R$ with $1<R\le 3/2$ and all $\epsilon>0$,
the competitive ratio of the general\&removable online knapsack problem with a buffer is at least $1/(R-1)-\epsilon$.
\end{theorem}
\begin{proof}
Let $\ALG$ be an online algorithm.
Let $\hat{\epsilon}$ be a positive real such that 
$1/\hat{\epsilon}$ is an integer and 
$\min\left\{\frac{1}{(R-1+\hat{\epsilon})(1+\hat{\epsilon})},\frac{1-\hat{\epsilon}^2}{R-1}\right\}\ge \frac{1}{R-1}-\epsilon$.
In addition, let $m\coloneqq 1/\hat{\epsilon}$ and $n\coloneqq 1/\hat{\epsilon}^3$.

Suppose that $\ALG$ is requested the following sequence of items:
\begin{align*}
(1,1),\ (\hat{\epsilon},\hat{\epsilon}^3),\ (\hat{\epsilon},2\hat{\epsilon}^3),\,\ldots\,,\ (\hat{\epsilon},n\hat{\epsilon}^3),
\end{align*}
until $\ALG$ discards the first item $(1,1)$.
Note that the first item has a large size and a medium density, and the following items have the same small sizes but different densities that slowly increase from small to large.
In addition, $\ALG$ must take the first item at the beginning (otherwise the competitive ratio becomes infinite).
Thus, $\ALG$ would keep the first item and the last $\lfloor\frac{R-1}{\hat{\epsilon}}\rfloor$ items in each round.

We have two cases to consider: $\ALG$ discards the first item $(1,1)$ or not.
\begin{description}
\item[Case 1:] Suppose that $\ALG$ discards the first item $(1,1)$ when the item $(\hat{\epsilon},i\hat{\epsilon}^3)$ comes.
Note that the requested sequence is $I\coloneqq\bigl((1,1),\,(\hat{\epsilon},\hat{\epsilon}^3),\,(\hat{\epsilon},2\hat{\epsilon}^3),\ldots,\,(\hat{\epsilon},i\hat{\epsilon}^3)\bigr)$.
Then, we have $\ALG(I)\le (\lfloor\frac{R-1}{\hat{\epsilon}}\rfloor+1) i\hat{\epsilon}^3$ (since $\ALG$ keeps at most $\lfloor\frac{R-1}{\hat{\epsilon}}\rfloor+1$ small items at the end) and $\OPT(I)\ge \max\{1,\, m\cdot (i-m)\hat{\epsilon}^3\}$ (the left term $1$ comes from the first item and the right term $m\cdot (i-m)\hat{\epsilon}^3$ comes from the last $m~(=1/\hat{\epsilon})$ items).
Hence, the competitive ratio is at least
\begin{align*}
\frac{\OPT(I)}{\ALG(I)}
&\ge \frac{\max\{1,\,m\cdot(i-m)\hat{\epsilon}^3\}}{(\lfloor\frac{R-1}{\hat{\epsilon}}\rfloor+1) i\hat{\epsilon}^3}\\
&\ge \frac{\max\{1,\,m\cdot(i-m)\hat{\epsilon}^3\}}{(\frac{R-1}{\hat{\epsilon}}+1) i\hat{\epsilon}^3}
= \max\left\{\frac{1}{(R-1+\hat{\epsilon})i\hat{\epsilon}^2},\frac{i-\frac{1}{\hat{\epsilon}}}{(R-1+\hat{\epsilon})i}\right\}\\
&\ge \frac{1}{(R-1+\hat{\epsilon})( \frac{1}{\hat{\epsilon}} + \frac{1}{\hat{\epsilon}^2})\hat{\epsilon}^2}
= \frac{1}{(R-1+\hat{\epsilon})(1+\hat{\epsilon})} \ge \frac{1}{R-1}-\epsilon,
\end{align*}
where the third inequality holds since 
the left term $\frac{1}{(R-1+\hat{\epsilon})i\hat{\epsilon}^2}$ is monotone decreasing in $i$,
the right term $\frac{i-\frac{1}{\hat{\epsilon}}}{(R-1+\hat{\epsilon})i}$ is monotone increasing in $i$,
and the two take the same value when $i=\frac{1}{\hat{\epsilon}}+\frac{1}{\hat{\epsilon}^2}$.
\item[Case 2:] Suppose that $\ALG$ does not reject the first item until the end.
Then, the competitive ratio is at least
\begin{align*}
\frac{\OPT(I)}{\ALG(I)}
&\ge \frac{m\cdot (n-m)\hat{\epsilon}^3}{\lfloor \frac{R-1}{\hat{\epsilon}}\rfloor n\hat{\epsilon}^3}
\ge \frac{1}{R-1}\cdot \frac{m(n-m)\hat{\epsilon}^3}{n\hat{\epsilon}^2}
= \frac{1-\hat{\epsilon}^2}{R-1} \ge \frac{1}{R-1}-\epsilon.
\qedhere
\end{align*}
\end{description}
\end{proof}

\begin{theorem}\label{thm:general/removable_lower_3/2<R<2}
For all $R$ with $3/2\le R< 2$ and all $\epsilon>0$,
the competitive ratio of the general\&removable online knapsack problem with a buffer is at least $2-\epsilon$.
\end{theorem}
\begin{proof}
Let $k$ be an integer such that $k>\max\{\frac{1}{2-R},\,\frac{1}{\epsilon}\}$.
Let $\ALG$ be an online algorithm.

Consider the item sequence $I\coloneqq (e_1,\dots,e_{k})$ where $(s(e_i),v(e_i))=(1-\tfrac{i}{2k^2},1-\tfrac{i}{2k})$ for $i=1,\dots,k$.
Then, at the end of the sequence, $\ALG$ must keep exactly one item
because it must select at least one item (otherwise the competitive ratio is unbounded)
and every pair of items exceeds the capacity of the buffer (i.e., $s(e_i)+s(e_j)\ge 2(1-\tfrac{k}{2k^2})=2-\tfrac{1}{k}>R$ for any $i,j\in\{1,\dots,k\}$).

Suppose that $\{e_i\}$ is selected in the buffer at the end of the sequence $I$.
If $i=k$, then the competitive ratio for $I$ is $\frac{\OPT(I)}{\ALG(I)}=\frac{v(e_1)}{v(e_k)}=\frac{1-\frac{1}{2k}}{1-\frac{1}{2}}=2-\frac{1}{k}>2-\epsilon$.
Otherwise (i.e., $i< k$), let us consider a sequence $I'\coloneqq(e_1,\dots,e_{k},e_{k+1})$ with $(s(e_{k+1}),v(e_{k+1}))=(\frac{i+1}{2k^2},1-\frac{i}{2k})$.
Then, the competitive ratio for $I'$ is at least
\[\frac{v(e_{i-1})+v(e_{k+2})}{v(e_{i})}=\frac{(1-\frac{i-1}{2k})+(1-\frac{i+1}{2k})}{1-\frac{i-1}{2k}}=2-\frac{1}{2k-i}\ge 2-\frac{1}{2k}>2-\epsilon.\qedhere\]
\end{proof}

\subsection{Upper bounds}
Here, we provide an asymptotically nearly optimal algorithm for large $R$
and an optimal algorithm for small $R$ ($< 2$).

First, we provide a $(1+O(\log R/R))$-competitive algorithm for the asymptotic case.
Suppose that $R$ is sufficiently large.
Let $m \coloneqq \lfloor (R-3)/2\rfloor$ and let $\epsilon~(\le 1)$ be a positive real such that $\log_{1+\epsilon} (1/\epsilon) = m$.
Note that we have $m=\Theta(\frac{1}{\epsilon}\log\frac{1}{\epsilon})$ and $\epsilon = O(\log R/R)$ (see Lemma~\ref{lem:relation_m_eps_R} in Appendix~\ref{sec:relation_m_eps_R}).

We partition all the items as follows.
Let $S$ be the set of items with size at most $\epsilon$.
Let $M$ be the set of items not in $S$ and let $M^j$ $(j\in \mathbb{Z})$ be the set of items $e \in M$ with $(1+\epsilon)^j\le v(e) < (1+\epsilon)^{j+1}$ (note that $j$ is not restricted to be positive).
Let us consider Algorithm~\ref{alg:weighted_general} for the problem.
Intuitively, the algorithm selects items in greedy ways for $S$ and each $M^j$ with $\nu_i\le j\le \mu_i$.
Note that for any $i\ge 1$, we have $\mu_i-\nu_i = 2m$.
For each $i\ge 1$, since $s(B_i\cap S)\le 2+\epsilon $ and $s(B_i\cap M^j) \le 1$ for any $\nu_i\le j \le \mu_i$, we have
$s(B_i)\le 2m+2+\epsilon \le 2((R-3)/2)+2+\epsilon \le R$.
Thus, the algorithm is applicable.

\begin{algorithm}[htb]
\caption{$(1+O(\log R/R))$-competitive algorithm}\label{alg:weighted_general}
$B_0\ot\emptyset$\;
\For{\(i\ot 1,2,\dots\)}{%
    $B_i \ot \emptyset$ and $B_i'\ot (B_{i-1}\cup\{e_i\})$\;
    \ForEach{$e\in B_i' \cap S$ in the non-increasing order of the density}{
    	$B_i \ot B_i \cup \{e\} $\;
    	\lIf{$s(B_i) >2 $}{\textbf{break}}
    }
	Let $e^*_i \in \argmax\{v(e) \mid e \in B_i'\}$\;
	Let $\mu_i \ot \lfloor \log_{1+\epsilon} v(e^*_i)\rfloor$ and $\nu_i \ot \lfloor \log_{1+\epsilon} \epsilon^2 v(e^*_i)\rfloor$%
        \tcp*{$e_i^*\in M^{\mu_i}$}
	\For{$j \ot \nu_i,\dots,\mu_i$}{
		\ForEach{$e\in B_i' \cap M^j$ in the non-decreasing order of the size}{
 			\lIf{$s(B_i\cap M^j) + s(e) \le 1$}{$B_i \ot B_i \cup \{e\}$}
		}
	}
}
\end{algorithm}

\begin{theorem}\label{thm:general_removable_general}
  Algorithm~\ref{alg:weighted_general} is $(1+O(\log R/R))$-competitive 
  for the general\&removable online knapsack problem with a buffer 
  when $R$ is a sufficiently large real. 
\end{theorem}
Let $I\coloneqq(e_1,\dots,e_n)$ be an input sequence, 
$B_{\OPT}\in\argmax\{v(X)\mid s(X)\le 1,~X\subseteq \{e_1,\dots,e_n\}\}$ be the offline optimal solution, and 
$B_{\ALG}\in\argmax\{v(X)\mid s(X)\le 1,~X\subseteq B_n\}$ be the outcome solution of $\ALG$.
We construct another feasible solution $B^*$ from $B_n$ by Algorithm~\ref{alg:weighted_general_solution}.
Note that $v(B_{\ALG}) \ge v(B^*)$.

\begin{algorithm}[htb]
  \caption{Construct a feasible solution}\label{alg:weighted_general_solution}
  $B^*\ot B_n\cap B_{\OPT}$\;
  \For{$k\ot \nu_n,\dots,\mu_n$}{
  	Let $r_k \ot |(B_{\OPT}\setminus B^*)\cap M^k|$\;
  	\For{$j\ot 1,\dots,r_k$}{
  		Let $a\ot\argmin\{s(e)\mid e\in (B_n\setminus B^*)\cap M^k\}$ and $B^*\ot B^*\cup\{a\}$\;
  	}
  }
  \While{$(B_n\setminus B^*)\cap S\ne\emptyset$}{
    Let $a\in\argmax\{v(e)/s(e)\mid e\in (B_n\setminus B^*)\cap S\}$\;
    \lIf{$s(B^*)+s(a)\le 1$}{$B^*\ot B^*\cup\{a\}$}
    \lElse{\textbf{break}}
  }
  \Return $B^*$\;
\end{algorithm}

To prove the theorem, we show the following two claims.
\begin{claim}\label{claim:general_removable_general1}
$v(B_{\OPT} \cap M) \le (1+\epsilon) v(B^*\cap M) + \epsilon v(B_{\OPT})$ and $s(B_{\OPT} \cap M) \ge s(B^*\cap M)$.
\end{claim}
\begin{claim}\label{claim:general_removable_general2}
$v(B_{\OPT} \cap S) \le v(B^*\cap S) + \epsilon(1+2\epsilon) v(B_{\OPT})$.
\end{claim}
With these claims, $B^*$ is feasible, and we have
\begin{align*}
  v(B_{\OPT}) &= v(B_{\OPT} \cap M) +  v(B_{\OPT} \cap S)\\
              &\le (1+\epsilon) v(B^*\cap M) + \epsilon v(B_{\OPT}) + v(B^*\cap S) + \epsilon(1+2\epsilon) v(B_{\OPT})\\
              &= (1+\epsilon) v(B^*) + (2\epsilon+2\epsilon^2)v(B_{\OPT}).
\end{align*}
This implies $(1-2\epsilon-2\epsilon^2)v(B_{\OPT}) \le (1+\epsilon)v(B^*)$.
Since $v(B^*)\le v(B_{\ALG})$, the competitive ratio of Algorithm~\ref{alg:weighted_general} is at most $\frac{1+\epsilon}{1-2\epsilon-2\epsilon^2} \le \frac{1+\epsilon}{1-3\epsilon} \le 1+ 6\epsilon= 1+ O(\log R/R)$, when $\epsilon<1/12$ (this inequality follows from the assumption that $R$ is sufficiently large).

The proof is completed by proving Claims~\ref{claim:general_removable_general1} and~\ref{claim:general_removable_general2}.
\begin{proof}[Proof of Claim~\ref{claim:general_removable_general1}]
Note that
\begin{align*}
  v(B_{\OPT} \cap M) = \sum_{k<\nu_n}v(B_{\OPT}\cap M^k) + \sum_{k\ge\nu_n}v(B_{\OPT}\cap M^k).
\end{align*}
For $e\in M^k$ with $k<\nu_n$, we have $s(e)>\epsilon$ and $v(e)< (1+\epsilon)^{\nu_n}\le \epsilon^2 v(e_n^*)$, and hence $v(e)/s(e) \le \epsilon^2v(e^*_n)/\epsilon \le \epsilon v(B_{\OPT})$. 
Thus, we have $\sum_{k<\nu_n}v(B_{\OPT}\cap M^k) \le \epsilon v(B_{\OPT})$.
For $k$ with $\mu_n\le k\le \nu_n$, the set $B_n\cap M^k$ is the greedy solution for $M^k$ according to the non-decreasing order of their size.
Hence, by the construction of $B^*$, the number of items in $B_{\OPT}\cap M^k$ equals to the number of items in $B^*\cap M^k$, and we have $s(B_{\OPT}\cap M^k)\ge s(B^*\cap M^k)$.
Also, for each $e \in B_{\OPT} \cap M^k$ and $f \in B^* \cap M^k$, $v(e)/v(f) < (1+\epsilon)^{k+1}/ (1+\epsilon)^k = (1+\epsilon)$.
Hence, $\sum_{k\ge\nu_n}v(B_{\OPT}\cap M^k) \le (1+\epsilon) \sum_{k\ge\nu_n}v(B^*\cap M^k)$.
\end{proof}

\begin{proof}[Proof of Claim~\ref{claim:general_removable_general2}]
It is sufficient to consider the case $B_{\OPT} \cap S \not\subseteq B_n$, since otherwise $B_{\OPT} \cap S \subseteq B^*\cap S$ and the claim clearly holds.
Hence, we have $s(B_n\cap S)>2$.
Let $B_n \cap S = \{f_1,f_2,\dots,f_{|B_n \cap S|}\}$ be sorted in non-increasing order of their density.
Let $f_j$ be the item with the largest index in $(B_n\cap S) \setminus B_{\OPT}$.
Also let $\ell\ge 1$ be the index such that $\sum_{i=1}^{\ell}s(f_i) \le 1 < \sum_{i=1}^{\ell+1}s(f_i)$.
There are two cases to consider: $j\le \ell$ and $j>\ell$.
\begin{description}
\item[Case 1:] Suppose that $j\le \ell$.
Then, by the definition of $f_j$, we have $ \{f_{\ell+1},\dots f_{|B_n \cap S|}\} \subseteq B_{\OPT}$.
Since $s(B_n\cap S) > 2$, we have $s(B_{\OPT}) \ge s(B_{\OPT}\cap S) \ge s(B_n \cap S) - \sum_{i=1}^{\ell}s(f_i) > 1$, which contradicts with $s(B_{\OPT}) \le 1$.

\item[Case 2:] Suppose that $j> \ell$.
In this case, we prove that $v(f_j) \le \epsilon(1+2\epsilon)v(B_{\OPT})$.
Since $\sum_{i=1}^{\ell}s(f_i) \ge 1-\epsilon$, we have 
  \[(1-\epsilon)\cdot \frac{v(f_j)}{s(f_j)} \le \sum_{i=1}^\ell s(f_i)\cdot \frac{v(f_j)}{s(f_j)}\le \sum_{i=1}^\ell s(f_i)\cdot \frac{v(f_i)}{s(f_i)} = \sum_{i=1}^{\ell}v(f_i) \le v(B_{\OPT}).\]
Therefore, $v(f_j) \le \frac{s(f_j)}{1-\epsilon} v(B_{\OPT}) \le \frac{\epsilon}{1-\epsilon} v(B_{\OPT}) \le \epsilon(1+2\epsilon) v(B_{\OPT})$ when $\epsilon \le 1/2$.
    
Since $s(B^*\cap M)\le s(B_{\OPT}\cap M)$ by construction of $B^*$, we have $s(B^* \cap S) + s(f_j) \ge s(B_{\OPT}\cap S)$.
By construction of $B_n\cap S$, we have $\min\{v(f)/s(f)\mid f\in (B_n\cap S)\setminus B_{\OPT}\} \ge \max\{v(f)/s(f)\mid f\in (B_{\OPT}\cap S)\setminus B_n\}$.
Therefore, $v(B^* \cap S) + v(f_j) \ge v(B_{\OPT}\cap S)$.
Moreover we have $v(f_j) \le \epsilon(1+2\epsilon) v(B_{\OPT})$, and the claim follows.
\qedhere
\end{description}
\end{proof}
The proof of Theorem~\ref{thm:general_removable_general} is completed.

Next, let us consider an algorithm that selects items according to the non-increasing order of the density.
The algorithm is formally described in Algorithm~\ref{alg:weighted_removable_to_2}.
We prove that it is optimal when $1<R< 2$.
\begin{algorithm}[htb]
  \caption{$\max\{1/(R-1),2\}$-competitive algorithm for $1<R< 2$}
  \label{alg:weighted_removable_to_2}
  $B_0\ot\emptyset$\;
  \For{\(i\ot 1,2,\dots\)}{%
    $B_i\ot \emptyset$\;
    \ForEach{$e\in B_{i-1}\cup\{e_i\}$ in the non-increasing order of the density}{
        \lIf{$s(B_i)+s(e)\le R$}{
            $B_i\ot B_i\cup\{e\}$
        }
    }
  }
\end{algorithm}

\begin{theorem}\label{thm:weighted_removable_to_2}
  Algorithm~\ref{alg:weighted_removable_to_2} is $\max\{1/(R-1),2\}$-competitive 
  for the general\&removable online knapsack problem with a buffer 
  when $1<R< 2$.
\end{theorem}
\begin{proof}
	Let $I\coloneqq(e_1,\dots,e_n)$ be an input sequence.
	Without loss of generality, we can assume that $\sum_{i=1}^n s(e_i)>R$ since otherwise $\ALG(I) = \OPT(I)$.
	Let $f_1,\dots,f_n$ be the rearrangement of $I$ according to the non-increasing order of the density, i.e., $\{f_1,\dots,f_n\}=\{e_1,\dots,e_n\}$ and $v(f_1)/s(f_1)\ge\dots\ge v(f_n)/s(f_n)$.
	Let $k~(\le n-1)$ be the index such that $\sum^{k}_{i=1}s(f_i)\le 1 < \sum^{k+1}_{i=1}s(f_i)$.
	Then, by the definition of the algorithm, we have $\{f_1,\dots,f_k\}\subseteq B_n$.
	There are two cases to consider: $f_{k+1}\not\in B_n$ and $f_{k+1}\in B_n$.
\begin{description}
\item[Case 1:] Suppose that $f_{k+1}\not\in B_n$.
Then, we have $\sum_{i=1}^{k+1}s(f_i) > R$, and hence $\sum_{i=1}^k s(f_i)> R-s(f_{k+1})\ge R-1$ by $s(f_{k+1}) \le 1$. 
Thus, $\OPT(I)$ is at most $\ALG(I)/(R-1)$ and the competitive ratio is at most $1/(R-1)$.
	
\item[Case 2:] Suppose that $f_{k+1}\in B_n$.
By a similar analysis of the famous $2$-approximation algorithm for the offline knapsack problem, 
we have $\OPT(I)\le \sum^{k}_{i=1}v(f_i)+v(f_{k+1}) \le 2\cdot \max\{\sum^{k}_{i=1}v(f_i),v(f_{k+1})\}\le 2\cdot \ALG(I)$.
Thus, the competitive ratio is at most $2$. \qedhere
\end{description}
\end{proof}

\section{Proportional\&Removable Case}\label{section:unweighted removable case}
In this section, we consider the proportional\&removable case.
We consider the following four cases separately:
(i) $1\le R\le \frac{1+\sqrt{2}}{2}$, (ii) $2-\frac{\sqrt{2}}{2} \leq R \leq 17-9\sqrt{3}$, (iii) $2\sqrt{3}-2\le R\le 3/2$, and (iv) general $R$ (see Figure~\ref{fig:prop_removable}).
We remark that the competitive ratios for $\frac{1+\sqrt{2}}{2}\le R\le 2-\frac{\sqrt{2}}{2}$ (and $17-9\sqrt{3}\le R\le 2\sqrt{3}-2$) can be obtained
by considering the upper bound for $R=\frac{1+\sqrt{2}}{2}$ in case (i) ($R=17-9\sqrt{3}$ in case (ii)) and the lower bound for $R=2-\frac{\sqrt{2}}{2}$ in case (ii) ($R=2\sqrt{3}-2$ in case (iii)).

\begin{figure}[htb]
\centering
\begin{tikzpicture}[domain=0:2,very thick,yscale=5,xscale=20]
    \draw[->] (0.95,1.1) -- (1.54,1.1) node[right] {$R$};
    \draw[->] (0.97,1.02) -- (0.97,1.68) node[above,yshift=-1mm,font=\small] {Competitive ratio};

    \draw[densely dotted,blue] (1.0000,1.25) to node[below] {\scriptsize (i)}   (1.2070,1.25); 
    \draw[densely dotted,blue] (1.2928,1.25) to node[below] {\scriptsize (ii)}  (1.4115,1.25); 
    \draw[densely dotted,blue] (1.4641,1.25) to node[below] {\scriptsize (iii)} (1.5000,1.25); 

    \draw[dashed,color=gray,thin] (1.0000,1.6180)--(0.97,1.6180) node[left,font=\scriptsize,black] {$\frac{1+\sqrt{5}}{2}$};
    \draw[dashed,color=gray,thin] (1.2070,1.4140)--(0.97,1.4140) node[left,yshift=1mm,font=\scriptsize,black] {$\sqrt{2}$};
    \draw[dashed,color=gray,thin] (1.4115,1.3660)--(0.93,1.3660) node[left,font=\scriptsize,black] {$\frac{1+\sqrt 3}{2}$};
    \draw[dashed,color=gray,thin] (1.5000,1.3333)--(0.97,1.3333) node[left,yshift=-1mm,font=\scriptsize,black] {$\frac{4}{3}$};

    \draw[dashed,color=gray,thin] (1.2070,1.4140)--(1.2070,1.1) node[below,font=\scriptsize,black] {$\frac{1+\sqrt{2}}{2}$};
    \draw[dashed,color=gray,thin] (1.2928,1.4140)--(1.2928,1.1) node[below,font=\scriptsize,black] {$2-\frac{\sqrt 2}{2}$};
    \draw[dashed,color=gray,thin] (1.4115,1.3660)--(1.4115,1.1) node[below,font=\scriptsize,xshift=-3mm,black] {$17{-}9\sqrt{3}$};
    \draw[dashed,color=gray,thin] (1.4641,1.3660)--(1.4641,1.1) node[below,font=\scriptsize,xshift=-1mm,black] {$2\sqrt{3}{-}2$};
    \draw[dashed,color=gray,thin] (1.5000,1.3330)--(1.5000,1.1) node[below,font=\scriptsize,black] {$\frac{3}{2}$};
    \draw[dashed,color=gray,thin] (1.0000,1.6180)--(1.0000,1.1) node[below,font=\scriptsize,black] {$1$};

    \draw (1.1,1.6) node[blue] {$\frac{1+\sqrt{4R+1}}{2R}$};
    \draw[blue,domain=1:1.207,smooth,variable=\x] plot ({\x},{(1+sqrt(1+4*\x))/(2*\x)});

    \draw[domain=1.207:1.2928,smooth,variable=\x] plot ({\x},{sqrt(2)});

    \draw (1.36,1.47) node[blue] {$\frac{\sqrt{16R+1}-1}{2R}$};
    \draw[blue,domain=1.2928:1.4115,smooth,variable=\x] plot ({\x},{(sqrt(1+16*\x)-1)/(2*\x)});

    \draw[domain=1.4115:1.4641,smooth,variable=\x] plot ({\x},{(1+sqrt(3))/2});

    \draw (1.485,1.43) node[blue] {$\frac{2}{R}$};
    \draw[blue,domain=1.4641:1.5,smooth,variable=\x] plot ({\x},{2/(\x)});
\end{tikzpicture}
\vskip -2mm
\caption{The competitive ratios for the proportional\&removable case with $1\le R\le \frac{3}{2}$.}\label{fig:prop_removable}
\end{figure}
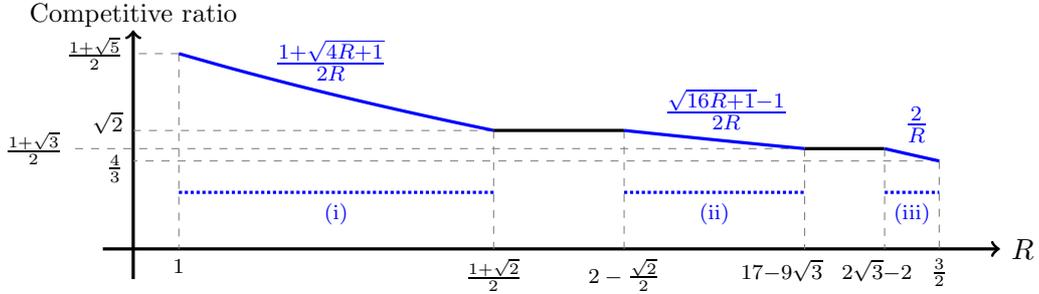

\subsection{\texorpdfstring{$1\le R\le \frac{1+\sqrt{2}}{2}$}{1le Rle frac{1+sqrt{2}}{2}}}
We prove that the competitive ratio is $\frac{1+\sqrt{4R+1}}{2R}$ when $1\le R\le \frac{1+\sqrt{2}}{2}$.
Let $r>0$ be a real such that $r+r^2=R$, i.e., $r=\frac{\sqrt{1+4R}-1}{2}$.

\subsubsection{Lower bound}
We first prove the lower bound.
\begin{theorem}\label{thm:lower_1leRle2}
  For any $\epsilon > 0$, 
  the competitive ratio of the proportional\&removable online knapsack problem with a buffer 
  is at least $\frac{1+\sqrt{4R+1}}{2R}-\epsilon$ when $1\le R< 2$.
\end{theorem}

\begin{proof}
  Let $\ALG$ be an online algorithm and 
  let $\epsilon'$ be a positive real such that $\frac{r}{r^2+\epsilon'}\ge \frac{1}{r}-\epsilon$ and $\epsilon'<r-r^2$.
  Note that $r = \frac{\sqrt{1+4R}-1}{2} < 1$ and $\frac{1}{r} = \frac{1+\sqrt{4R+1}}{2R}$.
  Consider the input sequence $I\coloneqq (e_1,e_2)$ where $s(e_1)=r$ and $s(e_2)=r^2+\epsilon'$.
  Since $r+r^2 = R$, $\ALG$ must discard at least one of them.
  If $\ALG$ discards the item with size $r$,  then the competitive ratio for the sequence is $\frac{r}{r^2+\epsilon'}\ge \frac{1}{r}-\epsilon=\frac{1+\sqrt{4R+1}}{2R}-\epsilon$.
  If $\ALG$ discards the item with size $r^2+\epsilon'$, 
  let $I'\coloneqq (e_1,e_2,e_3)$ where $s(e_3)=1-r^2-\epsilon'$. 
  As $r \ge 1-r^2$ and $r+(1-r^2-\epsilon')>1$, we have $\OPT(I')=1$ and $\ALG(I')\le r$.
  Hence the competitive ratio is at least $\frac{1}{r}=\frac{1+\sqrt{4R+1}}{2R}$.
\end{proof}

\subsubsection{Upper bound for \texorpdfstring{$1\le R\le 10/9$}{1le R le 10/9}}
Next, we give an optimal algorithm for $1\le R\le 10/9$.
In this subsubsection, an item \(e\) is called {\em small}, {\em medium}, and {\em large}
if \(s(e)\le r^2\), \(r^2<s(e)<r\), and \(r\le s(e)\), respectively.
Let $S$, $M$, and $L$ respectively denote the sets of small, medium, and large items%
(see Figure~\ref{fig:sizes1}).

\begin{figure}[htbp]
\begin{center}
  \begin{tikzpicture}[scale=0.8]
    \draw[thick,->] (-1,0)--(11,0);
    \filldraw[thick,fill opacity=0.1] (0,0)--(0,1)--(4.44,1)--(4.44,0); 
    \filldraw[thick,fill opacity=0.1] (4.44,0)--(4.94,1)--(6.16,1)--(6.66,0); 
    \filldraw[thick,fill opacity=0.1] (6.66,0)--(6.66,1)--(10,1)--(10,0); 
    \draw[thick] (0.00,-0.1)--(0.00,0.1);
    \draw[thick] (3.33,-0.1)--(3.33,0.1);
    \draw[thick] (3.6,-0.7)--(3.6,0.1);
    \draw[thick] (4.44,-0.1)--(4.44,0.1);
    \draw[thick] (5.00,-0.1)--(5.00,0.1);
    \draw[thick] (6.66,-0.1)--(6.66,0.1);
    \draw[thick] (10.00,-0.1)--(10.00,0.1);
    \draw (2,0.5) node {$S$};
    \draw (5.5,0.5) node {$M$};
    \draw (8.5,0.5) node {$L$};
    \draw (0,-0.5) node {$0$};
    \draw (3.33,-0.4) node {$\frac{r}{2}$};
    \draw (3.6,-1) node {$1-r$};
    \draw (4.44,-0.4) node {$r^2$};
    \draw (5.00,-0.4) node {$\frac{1}{2}$};
    \draw (6.66,-0.4) node {$r$};
    \draw (10.0,-0.4) node {$1$};
  \end{tikzpicture}
  \vskip -4mm
  \caption{Item partition for $1\le R\le 10/9$.}\label{fig:sizes1}
\end{center}
\end{figure}
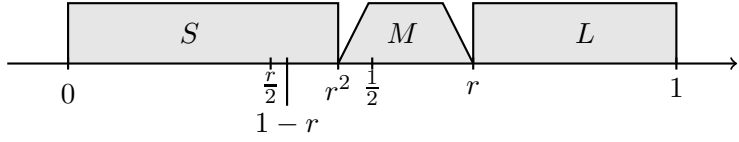

\begin{algorithm}[htb]
  \caption{$\frac{1+\sqrt{1+4R}}{2R}$-competitive algorithm for $1\le R\le \frac{10}{9}$}\label{alg:small}
  $B_0\ot\emptyset$\;
  \For{\(i\ot 1,2,\dots\)}{%
    \lIf{$\exists B'\subseteq B_{i-1}\cup\{e_i\}$ such that $r\le s(B')\le 1$}{$B_i\ot B'$}
    \ElseIf{$e_i\in M$ and $|B_{i-1}\cap M|=1$}{
      let $\{e_i'\}=B_{i-1}\cap M$\;
      \lIf{$s(e_i)<s(e_i')$}{$B_i\ot B_{i-1}\cup\{e_i\}\setminus \{e_i'\}$}
      \lElse{$B_i\ot B_{i-1}$}
    }
    \Else{
      $B_i\ot \emptyset$\;
      \ForEach{$e\in B_{i-1}\cup\{e_i\}$ in non-increasing order of size}{
        \lIf{$s(B_i)+s(e)\le R$}{$B_i\ot B_i\cup\{e\}$}
      }
    }
  }
\end{algorithm}

We consider Algorithm~\ref{alg:small},
which is a generalization of the $\frac{1+\sqrt{5}}{2}$-competitive algorithm for $R=1$ given by Iwama and Taketomi~\cite{IT2002}.
If the algorithm can select a set of items $B'$ such that $r\le s(B')\le 1$, it keeps the set $B'$ until the end since it is sufficient to achieve $1/r$-competitive.
Otherwise, it picks the smallest medium item (if exists) and greedily selects small items according to the non-increasing order of the sizes.
We show that it is optimal when $1\le R\le 10/9$.

\begin{theorem}\label{thm:proportional_removable_<10/9}
Algorithm~\ref{alg:small} is $\frac{1+\sqrt{1+4R}}{2R}$-competitive 
for the proportional\&removable online knapsack problem with a buffer 
when $1\le R\le 10/9$.
\end{theorem}
\begin{proof}
Let $I\coloneqq (e_1,\dots,e_n)$ be the input sequence.
If there exists a large item $e_i$, the competitive ratio is at most $1/r=\frac{1+\sqrt{1+4R}}{2R}$ by $r\le s(e_i)\le 1$.
If there exist two medium items $e_i, e_j$ such that $s(e_i)+s(e_j)\le 1$, the competitive ratio is at most $1/r=\frac{1+\sqrt{1+4R}}{2R}$ by $r<2r^2<s(e_i)+s(e_j)\le 1$.
In what follows, 
we assume that all the input items are not large and every pair of medium items cannot be packed into the knapsack together.
In addition, suppose that $s(B_n)\not\in[r,1]>1$ since otherwise the competitive ratio is at most $1/r=\frac{1+\sqrt{1+4R}}{2R}$.
By the algorithm, this additional assumption means $s(B')\not\in[1,r]$ for any $B'\subseteq B_{i-1}\cup\{e_i\}$ with $i\in\{1,\dots,n\}$.

If $\{e_1,\dots,e_n\}\cap S\subseteq B_n$, the competitive ratio is at most
\begin{align*}
  \frac{r+s(\{e_1,\dots,e_n\}\cap S)}{r^2+s(\{e_1,\dots,e_n\}\cap S)}\le \frac{1}{r}=\frac{1+\sqrt{1+4R}}{2R}.
\end{align*}
Otherwise, i.e., $\{e_1,\dots,e_n\}\cap S\not\subseteq B_n$,
let $e_i$ be a small item that is not in $B_n$,
and $j$ be the smallest index such that $j\ge i$ and $e_i\not\in B_j$.
Note that $e_i\in B_{j-1}\cup\{e_j\}$.
We have four cases to consider.
\begin{description}
\item[Case 1:] Suppose that $s(e_i)\ge r/2$. In this case, there exists $e'\in B_j$ such that $r^2\ge s(e')\ge s(e_i)$.
Thus, we have $r\le s(e_i)+s(e')\le 1$, a contradiction.

\item[Case 2:] Suppose that there exists no medium item in $B_j$.
Then, there exists $B'\subseteq B_{j-1}\cup\{e_j\}$ such that $r\le s(B_j)\le 1$, because $s(B_{j-1}\cup\{e_j\})>R$ and all the items in $B_{j-1}\cup\{e_j\}$ are small. 
This is a contradiction.

\item[Case 3:] Suppose that $s(e)<r/2$ for any $e\in B_j\cap S$.
Then, we have $r\le s(B_j)\le 1$, a contradiction.

\item[Case 4:] Let us consider the other case, i.e., 
$s(e_i)<r/2$, $\exists e\in B_j\cap M$, and $\exists e'\in B_j\cap S$ such that $s(e')\ge r/2$.
Then, $s(B_j)-s(e)+s(e_i)\ge R-s(e)\ge R-r^2=r$.
Also, $s(B_j)-s(e')\le R-s(e')\le R-r/2=r^2+r/2\le 1$.
By the additional assumption, we have $s(B_j)-s(e)+s(e_i)>1$ and $s(B_j)-s(e')<r$.
Thus, we have 
\begin{align*}
s(B_j)
>1+s(e)-s(e_i)
>1+r^2-r/2
=(1-r)^2+3r/2
\ge r+r/2\ge r+s(e')
>s(B_j),
\end{align*}
which is a contradiction. \qedhere
\end{description}
\end{proof}

\subsubsection{Upper bound for \texorpdfstring{$\frac{10}{9}\le R\le \frac{1+\sqrt{2}}{2}$}{frac{10}{9}le Rle frac{1+sqrt{2}}{2}}}
Recall that $r>0$ is a real such that $r+r^2=R$, i.e., $r=\frac{\sqrt{1+4R}-1}{2}$.
For $\frac{10}{9}\le R\le \frac{1+\sqrt{2}}{2}$, we have $2/3\le r\le 1/\sqrt{2}$ and $1-r\le r/2\le r^2\le 1/2<r<1$.
In this subsubsection, an item \(e\) is called {\em small}, {\em medium}, and {\em large}
if \(s(e)\le 1-r\), \(1-r<s(e)<r\), and \(r\le s(e)\), respectively.
Let $S$, $M$, and $L$ respectively denote the sets of small, medium, and large items.
In addition, $M$ is further partitioned into three subsets $M_i$ $(i=1,2,3)$, 
where $M_1$, $M_2$, $M_3$ respectively denote the set of the items $e$ with size $1-r<s(e)\le r/2$, $r/2< s(e)<r^2$, and $r^2\le s(e)<r$
(see Figure~\ref{fig:sizes2}).
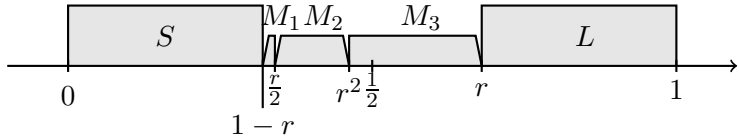
\begin{figure}[htbp]
\begin{center}
  \begin{tikzpicture}[scale=0.8]
    \draw[thick,->] (-1,0)--(11,0);
    \filldraw[thick,fill opacity=0.1] (0,0)--(0,1)--(3.20,1)--(3.20,0);
    \filldraw[thick,fill opacity=0.1] (3.20,0)--(3.30,0.5)--(3.40,0.5)--(3.40,0);
    \filldraw[thick,fill opacity=0.1] (3.40,0)--(3.50,0.5)--(4.52,0.5)--(4.62,0);
    \filldraw[thick,fill opacity=0.1] (4.62,0)--(4.62,0.5)--(6.70,0.5)--(6.80,0);
    \filldraw[thick,fill opacity=0.1] (6.80,0)--(6.80,1)--(10,1)--(10,0);
    \draw[thick] (0.00,-0.1)--(0.00,0.1);
    \draw[thick] (3.20,-0.7)--(3.20,0.1);
    \draw[thick] (3.40,-0.1)--(3.40,0.1);
    \draw[thick] (4.62,-0.1)--(4.62,0.1);
    \draw[thick] (5.00,-0.1)--(5.00,0.1);
    \draw[thick] (6.80,-0.1)--(6.80,0.1);
    \draw[thick] (10.00,-0.1)--(10.00,0.1);
    \draw (1.6,0.5) node {$S$};
    \draw (3.5,0.8) node {$M_1$};
    \draw (4.2,0.8) node {$M_2$};
    \draw (5.8,0.8) node {$M_3$};
    \draw (8.5,0.5) node {$L$};
    \draw (0,-0.5) node {$0$};
    \draw (3.40,-0.4) node {$\frac{r}{2}$};
    \draw (3.20,-1) node {$1-r$};
    \draw (4.62,-0.4) node {$r^2$};
    \draw (5.00,-0.4) node {$\frac{1}{2}$};
    \draw (6.80,-0.4) node {$r$};
    \draw (10.0,-0.4) node {$1$};
  \end{tikzpicture}
  \vskip -4mm
  \caption{Item partition for $10/9\le R\le \frac{1+\sqrt{2}}{2}$.}\label{fig:sizes2}
\end{center}
\end{figure}

We consider Algorithm~\ref{alg:small2} for the problem.
If the algorithm can select a set of items $B'$ such that $r\le s(B')\le 1$, it keeps the set $B'$ until the end.
Otherwise, it partitions the buffer into two spaces with size $r$ and $r^2$.
All the small items are taken into the first space.
If the set of medium items is of size at least $r^2$, then the smallest its subset $B'$ with size at least $r^2$ is selected into the first space.
If the set of medium items is of size at most $r^2$, then all of them are selected into the first space.
If there are remaining medium items, the smallest one is kept in the second space if its size is smaller than $r^2$.
We show that the algorithm is optimal when $\frac{10}{9}\le R\le \frac{1+\sqrt{2}}{2}$.

\begin{algorithm}[htb]
  \caption{$\frac{1+\sqrt{1+4R}}{2R}$-competitive algorithm for $\frac{10}{9}\le R\le \frac{1+\sqrt{2}}{2}$}\label{alg:small2}
  $B_0\ot\emptyset$, $B_0^{(1)}\ot\emptyset$, $B_0^{(2)}\ot\emptyset$\;
  \For{\(i\ot 1,2,\dots\)}{%
    \lIf{$\exists B'\subseteq B_{i-1}\cup\{e_i\}$ such that $r\le s(B')\le 1$\label{line:small2_c1}}{
      $B_i^{(1)}\ot B'$ and $B_i^{(2)}\ot \emptyset$
    }
    \ElseIf{$s((B_{i-1}\cup\{e_i\})\cap M)\ge r^2$\label{line:small2_c2}}{
      let $T_i\in\argmin\{s(B')\mid B'\subseteq (B_{i-1}\cup\{e_i\})\cap M,~s(B')\ge r^2\}$\;
      $B_i^{(1)}\ot T_i\cup ((B_{i-1}\cup\{e_i\})\cap S)$\;
      \If{$B_i^{(1)}\ne B_{i-1}\cup\{e_i\}$}{
        let $a\in\argmin\{s(e)\mid e\in B_{i-1}\cup\{e_i\}\setminus B_i^{(1)}\}$\;
        \lIf{$a\in M_1\cup M_2$}{$B_i^{(2)}\ot \{a\}$}
      }
    }
    \label{line:small2_c3}\lElse{
      $B_i^{(1)}\ot B_{i-1}\cup\{e_i\}$ and $B_i^{(2)}\ot \emptyset$
    }
    $B_i\ot B_i^{(1)}\cup B_i^{(2)}$\;
  }
\end{algorithm}

\begin{theorem}\label{thm:10/9_to_(1+sqrt 2)/2}
Algorithm~\ref{alg:small2} is $\frac{1+\sqrt{1+4R}}{2R}$-competitive 
for the proportional\&removable online knapsack problem with a buffer 
when $10/9\le R\le \frac{1+\sqrt{2}}{2}$.
\end{theorem}

Let $I\coloneqq(e_1,\dots,e_n)$ be the input sequence and
let $I_k\coloneqq\{e_1,\dots,e_k\}$ be the first $k$ items of $I$.

To prove the theorem, we first provide some observations.
\begin{observation}\label{obs:10/9_to_(1+sqrt 2)/2:valid_bin}
For each $i$, we have
(i) $s(B_i^{(1)})<r$ and $s(B_i^{(2)})\le r^2$ or
(ii) $r\le s(B_i^{(1)})\le 1$ and $s(B_i^{(2)})=0$.
\end{observation}
\begin{proof}
Let us consider round $i$.
If the condition in Line~\ref{line:small2_c1} is satisfied, we have $r\le s(B_i^{(1)})\le 1$ and $s(B_i^{(2)})=0$.

Suppose that the condition in Line~\ref{line:small2_c2} is satisfied, i.e., 
$s(B')\not\in[r,1]$ for any $B'\subseteq B_{i-1}\cup\{e_i\}$ and $s((B_{i-1}\cup\{e_i\})\cap M)\ge r^2$.
We have $s(B_i^{(2)})\le r^2$ since $B_i^{(2)}$ is the empty set or the singleton of an item in $M_1\cup M_2$.
Suppose to the contrary that $s(B_i^{(1)})\ge r$.
Then, as $s(B')\not\in[r,1]$ for any $B'\subseteq B_{i-1}\cup\{e_i\}$, we have $s(B_i^{(1)})>1$.
Thus, by removing some small items from $B_i^{(1)}$, we can obtain a subset with size in $[r,1]$, a contradiction

Suppose that the condition in Line~\ref{line:small2_c3} is satisfied, i.e.,
$s(B')\not\in[r,1]$ for any $B'\subseteq B_{i-1}\cup\{e_i\}$ and $s((B_{i-1}\cup\{e_i\})\cap M)< r^2$.
In this case, it is clear that $s(B_i^{(2)})=0$.
In addition, we have $s(B_i^{(1)})<r$, since otherwise we can obtain a subset with size in $[r,1]$ by a similar argument as above.
\end{proof}

\begin{observation}\label{obs:10/9_to_(1+sqrt 2)/2:M1M2}
If $a,b\in M_1\cup M_2$, then $r^2\le s(a)+s(b)\le 1$.
\end{observation}
\begin{proof}
  It is because
  $s(a)+s(b)\le \frac{1}{2}+\frac{1}{2}=1$ and 
  $s(a)+s(b)\ge 2(1-r)\ge r^2$ by $r\le \sqrt{2}/2\le \sqrt{3}-1$.
\end{proof}

\begin{observation}\label{obs:10/9_to_(1+sqrt 2)/2:M2M2}
If $a,b\in M_2$, then $r\le s(a)+s(b)\le 1$.
\end{observation}
\begin{proof}
  It is because
  $s(a)+s(b)\le \frac{1}{2}+\frac{1}{2}=1$ and
  $s(a)+s(b)\ge \frac{r}{2}+\frac{r}{2}=r$.
\end{proof}

\begin{observation}\label{obs:10/9_to_(1+sqrt 2)/2:M1M3}
If $a\in M_1$ and $b\in M_3$, then $s(a)+s(b)\ge r$.
\end{observation}
\begin{proof}
  It is because $s(a)+s(b)\ge (1-r)+r^2\ge r$.
\end{proof}

By the definition of the algorithm, we have the following two observations.
\begin{observation}\label{obs:10/9_to_(1+sqrt 2)/2:smallest_M-item}
  If $s(B_k^{(1)})<r$ and $I_k\cap M\ne\emptyset$,
  then $B_k\cap \argmin\{s(e)\mid e\in I_k\cap M\}\ne\emptyset$.
\end{observation}

\begin{observation}\label{obs:10/9_to_(1+sqrt 2)/2:smallest_M-set}
  If $s(B_k^{(1)})<r$ and $s(I_k\cap M)\ge r^2$,
  then $B_k\cap M\in\argmin\{s(B')\mid B'\subseteq I_k\cap M,~s(B')\ge r^2\}$.
\end{observation}

The following lemmas show the competitive ratio of Algorithm~\ref{alg:small2} 
when all the input items are medium, i.e., $I_n\subseteq M$.

\begin{lemma}\label{lem:10/9_to_(1+sqrt 2)/2:OPT1}
  If $I_n\subseteq M$ and $|\OPT| = 1$,
  then Algorithm~\ref{alg:small2} is $\frac{1+\sqrt{1+4R}}{2R}~(=1/r)$-competitive when $10/9\le R\le \frac{1+\sqrt{2}}{2}$.
\end{lemma}
\begin{proof}
  Let $\OPT=\{x\}$.
  If $x\in M_3$, the competitive ratio is at most $s(x)/s(B_n^{(1)})\le r/r^2=1/r$.
  If $x\in M_1\cup M_2$, then $|I|=1$ and the lemma holds by $s(x)\le 1/2$.
\end{proof}

\begin{lemma}\label{lem:10/9_to_(1+sqrt 2)/2:OPT3}
  If $I_n\subseteq M$ and $|\OPT| \ge 3$,
  then Algorithm~\ref{alg:small2} is $\frac{1+\sqrt{1+4R}}{2R}~(=1/r)$-competitive when $10/9\le R\le \frac{1+\sqrt{2}}{2}$.
\end{lemma}
\begin{proof}
  We have $|\OPT| = 3$ by $4(1-r)>1$.

  Let $\OPT=\{e_{i_1},e_{i_2},e_{i_3}\}$ where $i_1<i_2<i_3$.
  We assume $B_{n}^{(1)}< r$, since otherwise the lemma clearly holds.
  Since $r^2\le2(1-r)\le s(e_{i_1})+s(e_{i_2})$, we have $s(B^{(1)}_{i_3-1})\le s(e_{i_1})+s(e_{i_2})$. 
  Hence, we have
  \(1\ge s(e_{i_1})+s(e_{i_2})+s(e_{i_3})\ge s(B^{(1)}_{i_3-1})+s(e_{i_3})\ge r^2+(1-r)\ge r\),
  which implies the lemma.
\end{proof}

\begin{lemma}\label{lem:10/9_to_(1+sqrt 2)/2:OPT2_M3}
  If $I_n\subseteq M$, $|\OPT| = 2$, and $\OPT\cap M_3 \neq \emptyset$,
  then Algorithm~\ref{alg:small2} is $\frac{1+\sqrt{1+4R}}{2R}~(=1/r)$-competitive when $10/9\le R\le \frac{1+\sqrt{2}}{2}$.
\end{lemma}
\begin{proof}
  Let $\OPT=\{e_{i_1},e_{i_2}\}$ where $i_1<i_2$.
  We assume $B_{n}^{(1)}< r$, since otherwise the lemma clearly holds.
  If $e_{i_1}\in M_3$, then $r\le r^2+(1-r)\le s(B^{(1)}_{i_2-1})+s(e_{i_2})\le s(e_{i_1})+s(e_{i_2})\le 1$ and $\ALG(I)\ge r$.
  If $e_{i_2}\in M_3$, let $a\in B_{i_2}\cap\argmin\{s(e)\mid e\in I_{i_2-1}\}$.
  Then $r\le (1-r)+r^2\le s(a)+s(e_{i_2})\le s(e_{i_1})+s(e_{i_2})\le 1$ and $\ALG(I)\ge r$.
\end{proof}

\begin{lemma}\label{lem:10/9_to_(1+sqrt 2)/2:OPT2_min}
  Suppose that $I_n\subseteq M$ and $\OPT = \{e_{i_1},e_{i_2}\}\subseteq M_1\cup M_2$ where $i_1<i_2$.
  Let $e^*\in B_n\cap \argmin\{s(e)\mid e\in I\}$.
  If $s(e^*) + s(e_{i_1}) \ge r$ or $s(e^*) + s(e_{i_2}) \ge r$, then the Algorithm~\ref{alg:small2} is $\frac{1+\sqrt{1+4R}}{2R}~(=1/r)$-competitive when $10/9\le R\le \frac{1+\sqrt{2}}{2}$.
\end{lemma}
\begin{proof}
  Note that $e^* \in M_1\cup M_2$.

  Suppose that $s(e^*) + s(e_{i_1}) \ge r$.
  If $I_{i_1-1} \cap (M_1\cup M_2) \neq \emptyset$, let $f \in \argmin\{s(e') \mid e'\in I_{i_1-1} \cap (M_1\cup M_2)\}$.
  By observation~\ref{obs:10/9_to_(1+sqrt 2)/2:smallest_M-item}, $f\in B_{i_1-1}$.
  Then, when $e_{i_1}$ is given, $1 > s(f) + s(e_{i_1})\ge s(e^*) + s(e_{i_1}) \ge r$.
  Therefore $\ALG(I)\ge r$.
  If $I_{i_1-1} \cap (M_1\cup M_2) = \emptyset$, $\argmin\{s(e') \mid e'\in I_{i_1} \cap (M_1\cup M_2)\} = \{e_{i_1}\}$.
  Let $j>i_1$ be the first round such that $e_j \in M_1\cup M_2$.
  Then,$1 > s(e_j) + s(e_{i_1})\ge s(e^*) + s(e_{i_1}) \ge r$.
  Therefore $\ALG(I)\ge r$.

  Suppose that $s(e^*) + s(e_{i_2}) \ge r$.
  Since $e_{i_1} \in M_1 \cup M_2$, there exists $f\in B_{i_2-1}$ such that $s(f)\le s(e_{i_1})$.
  Then $1 > s(f) + s(e_{i_2})\ge s(e^*) + s(e_{i_2}) \ge r$.
  Therefore $\ALG(I)\ge r$.
\end{proof}

\begin{lemma}\label{lem:10/9_to_(1+sqrt 2)/2:OPT2_1}
  If $I_n\subseteq M$, $|\OPT| = 2$, $\OPT\cap M_3 = \emptyset$ and $|B_n^{(1)}|=1$,
  then Algorithm~\ref{alg:small2} is $\frac{1+\sqrt{1+4R}}{2R}~(=1/r)$-competitive when $10/9\le R\le \frac{1+\sqrt{2}}{2}$.
\end{lemma}
\begin{proof}
  Let $\OPT=\{e_{i_1},e_{i_2}\}$ where $i_1<i_2$.
  Let $y \in B_n^{(1)}$ and $e^*\in B_n\cap \argmin\{s(e)\mid e\in I\}$.
  Note that $y\in M_3$ and $y\not\in \OPT$.
  As $s(e^*)+s(y)\ge (1-r)+r^2\ge r$, we assume that $s(e^*)+s(y)>1$.

  Suppose that $|\OPT\cap B_n|=0$. 
  We assume that $s(e^*)+s(e_{i_1})<r$ and $s(e^*)+s(e_{i_2})<r$ since otherwise $\ALG(I)\ge r$ by Lemma~\ref{lem:10/9_to_(1+sqrt 2)/2:OPT2_min}.
  Then, the competitive ratio is at most
  \(\frac{s(e_{i_1})+s(e_{i_2})}{s(y)}
        \le \frac{2(r-s(e^*))}{1-s(e^*)}
        \le 2+\frac{2(r-1)}{1-s(e^*)}
        \le 2+\frac{2(r-1)}{1-r^2}
        \le \frac{1}{r}\)
  by $r<1$.

  Suppose that $|\OPT\cap B_n|\ge 1$.
  In this case, $e^*\in\{e_{i_1},e_{i_2}\}$.
  Let $\OPT=\{e^*,z\}$.
  We assume that $s(e^*)+s(z)<r$ by Lemma~\ref{lem:10/9_to_(1+sqrt 2)/2:OPT2_min}.
  Then, $s(e^*)+s(z)< s(e^*)+(r-s(e^*))=r$ and $s(B_n^{(1)}) = s(y)\ge r^2$.
  Hence, the competitive ratio is at most $1/r$.
\end{proof}

\begin{lemma}\label{lem:10/9_to_(1+sqrt 2)/2:OPT2_2}
  If $I_n\subseteq M$, $|\OPT| = 2$, $\OPT\cap M_3 = \emptyset$ and $|B_n^{(1)}|=2$,
  then Algorithm~\ref{alg:small2} is $\frac{1+\sqrt{1+4R}}{2R}~(=1/r)$-competitive when $10/9\le R\le \frac{1+\sqrt{2}}{2}$.
\end{lemma}
\begin{proof}
  Let $s(B_n^{(1)})=\{y_1,y_2\}$ $(s(y_1)\le s(y_2))$ and $e^*\in B_n\cap \argmin\{s(e)\mid e\in I\}$.
  Note that $y_1=e^*\in M_1\cup M_2$.
  We have five cases to consider.
\begin{description}
\item[Case 1:] $y_2\in M_3$. In this case, the lemma holds by $1\ge s(y_1)+s(y_2)\ge (1-r)+r^2\ge r$.

\item[Case 2:] $y_1,y_2\in M_1\cup M_2$ and $|B_n^{(1)}\cap\OPT|=0$.
    We assume that $s(y_1)+s(e_{i_1})<r$ and $s(y_1)+s(e_{i_2})<r$ since otherwise $\ALG(I)\ge r$ by Lemma~\ref{lem:10/9_to_(1+sqrt 2)/2:OPT2_min}.
    Thus, the competitive ratio is at most
    \(\frac{s(e_{i_1})+s(e_{i_2})}{s(y_1)+s(y_2)}\le\frac{2(r-s(e^*))}{2s(e^*)}\le\frac{2r-1}{1-r}\le\frac{1}{r}\),
    by $r\le 1/\sqrt{2}$.

\item[Case 3:] $y_1,y_2\in M_1\cup M_2$, $|B_n^{(1)}\cap\OPT|=1$, and $y_1\in\OPT$.
    Let $\OPT = \{y_1,z\}$.
    We can assume that $s(y_1)+s(z)<r$ by Lemma~\ref{lem:10/9_to_(1+sqrt 2)/2:OPT2_min}.
    Then, the competitive ratio is at most
    \(\frac{s(y_1)+s(z)}{s(y_1)+s(y_2)}\le\frac{s(y_1)+(r-s(y_1))}{s(e^*)+s(y_2)}\le\frac{r}{s(y_1)+s(y_2)}\le \frac{r}{r^2} = \frac{1}{r}\).

\item[Case 4:] $y_1,y_2\in M_1\cup M_2$, $|B_n^{(1)}\cap\OPT|=1$, and $y_2\in\OPT$. 
      Let $\OPT=\{y_2,z\}$.
      We assume that $s(z)+s(e^*)<r$ by Lemma~\ref{lem:10/9_to_(1+sqrt 2)/2:OPT2_min} and the competitive ratio is at most
      \(\frac{s(z)+s(y_2)}{s(y_1)+s(y_2)}\le\frac{(r-s(e^*)+s(y_2))}{s(e^*)+s(y_2)}\le\frac{r-s(e^*)}{s(e^*)}\le \frac{r}{1-r} - 1\le\frac{1}{r}\).

\item[Case 5:] $y_1,y_2\in M_1\cup M_2$ and $|B_n^{(1)}\cap\OPT|=2$. We have the lemma by $\OPT=B_n^{(1)}$. \qedhere
\end{description}
\end{proof}

Now, we are ready to prove Theorem~\ref{thm:10/9_to_(1+sqrt 2)/2}.
\begin{proof}[Proof of Theorem~\ref{thm:10/9_to_(1+sqrt 2)/2}]
  Let $\OPT\in\argmax\{s(X)\mid X\subseteq I_n,~s(X)\le 1\}$ and 
  $\OPT_M\in\argmax\{s(X)\mid X\subseteq I_n\cap M,~s(X)\le 1\}$.
  Without loss of generality, we can assume that $\sum_{i=1}^n s(e_i)>R$.

  If $e_i\in L$ for some $i$, then $r\le s(B_n^{(1)})\le 1$.
  Thus, we assume that all the items in the input sequence are not large, i.e., $I_n\cap L=\emptyset$.

  Suppose that Algorithm~\ref{alg:small2} discards some small items, i.e., $I_n\cap S\neq B_n\cap S$.
  Let $j$ be the round such that $I_{j-1}\cap S = B_{j-1}\cap S$ and $I_j\cap S\neq B_j\cap S$.
  Let $T_j\in\argmin\{s(B')\mid B'\subseteq (B_{j-1}\cup\{e_j\})\cap M,~s(B')>r\}$.
  Since $I_{j-1}\cap S = B_{j-1}\cap S$ and $I_j\cap S\neq B_j\cap S$, we have $s(T_j \cup (I_j\cap S))>1$.
  Since $s(e)<1-r~(\forall e\in S)$, there exists $S' \in I_j \cap S$ such that $r\le s(T_j\cup S') \le 1$.
  Therefore, if $I_n\cap S\neq B_n\cap S$, then $\ALG(I) \ge r$.

  Consequently, we assume $I_n\cap L=\emptyset$ and $I_n\cap S\subseteq B_n$.
  Then, the competitive ratio is at most
  \begin{align*}
    \frac{s(\OPT)}{s(B_n^{(1)})}\le \frac{s(\OPT_M)+s(I_n\cap S)}{s(B_n^{(1)}\cap M)+s(I_n\cap S)}\le \frac{s(\OPT_M)}{s(B_n^{(1)}\cap M)},
  \end{align*}
  and hence we can assume, without loss of generality, that $I_n\subseteq M$.

  Thus, by Lemmas~\ref{lem:10/9_to_(1+sqrt 2)/2:OPT1}--\ref{lem:10/9_to_(1+sqrt 2)/2:OPT2_2}, the theorem is proved.
\end{proof}

\subsection{\texorpdfstring{$\frac{4-\sqrt{2}}{2} \leq R \leq 17-9\sqrt{3}$}{frac{4-sqrt{2}}{2} leq R leq 17-9sqrt{3}}}

In this subsection, we consider the problem for $\frac{4-\sqrt{2}}{2} \leq R \leq 17-9\sqrt{3}$.
Let $r > 0$ be a real such that $\frac{R-r}{2(2r-1)} = r$, i.e , $r = \frac{\sqrt{16R+1}+1}{8}$.
Note that $\frac{1}{r}=\frac{\sqrt{16R+1}-1}{2R}$,  $\frac{2}{3}<r<\frac{3}{4} $ , $2r>R$, $r\geq R-1$, and $r\geq 2-R$.
We prove that the competitive ratio is $\frac{\sqrt{16R+1}-1}{2R}$ when $\frac{4-\sqrt{2}}{2} \leq R \leq 17-9\sqrt{3}$.

\subsubsection{Lower bound}
First we give a lower bound of the competitive ratio.
\begin{theorem}\label{thm:prop/removable_lower_ii}
  For any $\epsilon > 0$, 
  the competitive ratio of the proportional\&removable online knapsack problem with a buffer
  is at least $\frac{\sqrt{16R+1}-1}{2R}-\epsilon$
  when $\frac{4-\sqrt{2}}{2}\leq R < 17-9\sqrt{3}$. 
\end{theorem}

\begin{proof}
  Let $\ALG$ be an online algorithm.
  Consider the following item sequence:
  \begin{align*}
    r,\, 1-r+\frac{\epsilon}{4},\, 2r-1.
  \end{align*}
  Here, $r+(1-r+\epsilon/4)+(2r-1) = 2r + \epsilon/4> R$. 
  Hence, $\ALG$ must remove at least one of the items.

  If $\ALG$ removes the first item, then we assume that the fourth item with size $1-r$ arrives.
  In this case, the competitive ratio of the algorithm is at least $\frac{r+(1-r)}{(1-r+\epsilon/4)+(2r-1)} =  \frac{1}{r+\epsilon/4}\ge \frac{1}{r}-\epsilon$.

  If $\ALG$ removes the second item, then we assume that the fourth item with size $r-\epsilon/4$ arrives.
  In this case, the competitive ratio is at least $\frac{(1-r+\epsilon/4)+(r-\epsilon/4)}{r}=\frac{1}{r}$.

  If $\ALG$ removes the third item, then we assume that the fourth item with size $2r-1$ arrives.
  Here, if it discards the first or second item, then the next item with size $1-r$ or $r-\epsilon/4$ implies that the competitive ratio is at least $\frac{1}{r}-\epsilon$ by the similar analysis to the above.
  Thus, suppose that it discards the fourth item. We assume that the fifth item has size $R-1+\epsilon/4$.
  Now, it keeps three items with size $r$, $1-r+\epsilon/4$, $R-1+\epsilon/2$,
  whose total size $r+(1-r+\epsilon/4)+(R-1+\epsilon/2)=R+3\epsilon/4$ is larger than $R$.
  Hence, it must discards one of them.
  If it discards the item with size $r$, then the competitive ratio is at least $\frac{2(2r-1)}{(R-1+\epsilon/2)+(1-r+\epsilon/4)}=\frac{2(2r-1)}{R-r+3\epsilon/4}=\frac{1}{r+3\epsilon/(8(2r-1))}\ge \frac{1}{r}-\epsilon$.f
  If it discards the item with size $1-r+\epsilon/4$, then the next item has size $r-\epsilon/4$ and the competitive ratio is at least $\frac{(1-r+\epsilon/4) + (r-\epsilon/4)}{r} \geq \frac{1}{r} -\epsilon $.If it discards the item with size $R-1$, then the next item has size $1+r-R-3\epsilon/4$ and the competitive ratio is at least $\frac{(1-r+\epsilon/4) + (R-1+\epsilon/2) + (1+r-R-3\epsilon/4)}{r} = \frac{1}{r}$.

  Therefore, any online algorithm has competitive ratio at least $\frac{1}{r}-\epsilon = \frac{\sqrt{16R+1}-1}{2R}-\epsilon$.
\end{proof}

\subsubsection{Upper bound}
In this subsubsection, an item \(e\) is called {\em small}, {\em medium}, and {\em large}
if \(s(e)\le 1-r\), \(1-r< s(e)<r\), and \(r\le s(e)\), respectively.
Let $S$, $M$, and $L$ respectively denote the sets of small, medium, and large items.
$M$ is further partitioned into four subsets $M_i$ for $ 1 \le i \le 4$,
nwhere $M_1$, $M_2$, $M_3$, and $M_4$
respectively denote the set of the items $e$ with size
\(1-r< s(e)< 2r-1\), \(2r-1\le s(e)< 1/2\), \(1/2\le s(e)<r^2\), and \(r^2\le s(e)<r\)
(see Figure~\ref{fig:17-9sqrt(3)}).
An item $e$ is also called an {\em $M_i$-item} if \(e\in M_i\).

\begin{figure}[htbp]
\begin{center}
  \begin{tikzpicture}[scale=0.8]
    \draw[thick,->] (-1,0)--(11,0);
    \filldraw[thick,fill=gray,fill opacity=0.1] (0,0)--(0,1)--(2.5,1)--(2.5,0);
    \filldraw[thick,fill=gray,fill opacity=0.1] (2.5,0)--(2.8,.5)--(3.45,.5)--(3.75,0);
    \filldraw[thick,fill=gray,fill opacity=0.1] (3.75,0)--(3.75,.5)--(4.7,.5)--(5,0);
    \filldraw[thick,fill=gray,fill opacity=0.1] (5,0)--(5,.5)--(5.325,.5)--(5.625,0);
    \filldraw[thick,fill=gray,fill opacity=0.1] (5.625,0)--(5.625,.5)--(7.2,.5)--(7.5,0);
    \filldraw[thick,fill=gray,fill opacity=0.1] (7.5,0)--(7.5,1)--(10,1)--(10,0);
    \draw (1.25,0.5) node {$S$};
    \draw (3.125,0.8) node {$M_1$};
    \draw (4.375,0.8) node {$M_2$};
    \draw (5.3125,0.8) node {$M_3$};
    \draw (6.5625,0.8) node {$M_4$};
    \draw (8.75,0.5) node {$L$};
    \draw[thick] (0.00,-0.1)--(0.00,0.1);
    \draw (0,-0.5) node {$0$};
    \draw[thick] (2.5,-0.1)--(2.5,0.1);
    \draw (2,-0.4) node {$1-r$};
    \draw[thick] (3.75,-0.1)--(3.75,0.1);
    \draw (3.75,-0.4) node {$2r-1$};
    \draw[thick] (5,-0.1)--(5,0.1);
    \draw (5,-0.6) node {$\frac{1}{2}$};
    \draw[thick] (5.625,-0.1)--(5.625,0.1);
    \draw (5.625,-0.4) node {$r^2$};
    \draw[thick] (7.5,-0.1)--(7.5,0.1);
    \draw (7.5,-0.4) node {$r$};
    \draw[thick] (10.00,-0.1)--(10.00,0.1);
    \draw (10.0,-0.4) node {$1$};
  \end{tikzpicture}
  \caption{Item partition for $2-\sqrt{2}/2\le R\le 17-9\sqrt{3}$.}\label{fig:17-9sqrt(3)}
\end{center}
\end{figure}
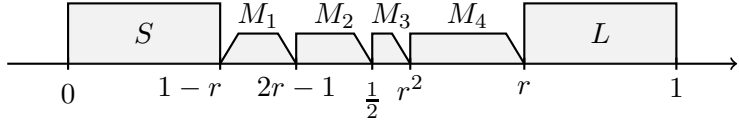

We consider Algorithm~\ref{alg:17-9sqrt(3)} for the problem.
If the algorithm can select a set of items $B'$ such that $r\le s(B')\le 1$, it keeps the set $B'$ until the end.
Otherwise, it first selects the smallest medium item.
If the total size of the two smallest items is smaller than the size of the smallest $M_4$-item, then it picks the second smallest $M_1$-item.
If the total size of the two smallest items is not smaller than the size of the smallest $M_4$-item, then it picks the smallest $M_4$-item.
If the remaining smallest medium item and the selected items can be taken into two bins with sizes $r$ and $R-r$, then it picks the item.
All the small items are taken into the remaining space.
We show that the algorithm is optimal when $\frac{4-\sqrt{2}}{2}\le R\le 17-9\sqrt{3}$.

\begin{algorithm}[htb]
  \caption{$\frac{\sqrt{16R+1}-1}{2R} (=\frac{1}{r})$-competitive algorithm}\label{alg:17-9sqrt(3)}
  $B_0\ot\emptyset$\;
  \For{\(i\ot 1,2,\dots\)}{%
    \lIf{$\exists B'\subseteq B_{i-1}\cup\{e_i\}$ such that $r\le s(B')\le 1$}{$B_i\ot B'$}\label{line:17-9sqrt(3)_win}
    \Else{
      let $B_i\ot \emptyset$, $A_j\ot (B_{i-1}\cup\{e_i\})\cap M_j~(j=1,2,3,4)$\;
      \If{$|A_1\cup A_2\cup A_3|>0$}{
        let $a\in\argmin_{e\in A_1\cup A_2\cup A_3}s(e)$ and $B_i\ot \{a\}$\;
      }
      \If{$|A_1|\ge 2$}{
        let $b\in\argmin_{e\in A_1\setminus B_i}s(e)$\;
        \lIf{$s(B_i)+s(b)\le\min_{e\in A_4}s(e)$}{$B_i\ot B_i\cup\{b\}$}
      }
      \If{$|A_4|>0$ and $|B_i\cap M_1|\le 1$}{
        let $a\in\argmin_{e\in A_4}s(e)$ and $B_i\ot B_i\cup\{a\}$\;
      }
     \If{$(|B_i \cap M_1| = 2 \text{ and } s(B_i \cap M_1)\le R- r)$ or $(|B_i \cap M_4| = 1\text{ and }s(B_i \cap M_4)\le R- r)$}{
        let $a \in \argmin_{e \in (\cup_{j=1}^{4} A_j) \setminus B_i}s(e)$\;
        \lIf{$s(B_i)+s(a) \le R$}{$B_i\ot B_i\cup\{a\}$}
      }
      $B_i\ot B_i\cup ((B_{i-1}\cup\{e_i\})\cap S)$\;
    }
  }
\end{algorithm}


\begin{theorem}\label{thm:17-9sqrt(3)}
  Algorithm~\ref{alg:17-9sqrt(3)} is $\frac{\sqrt{16R+1}-1}{2R} (=1/r)$-competitive when $\frac{1+\sqrt{2}}{2}\le R \le 17-9\sqrt{3}$.
\end{theorem}

Let $I = (e_1,\dots,e_n)$ be the input sequence, 
$I_k=\{e_1,\dots,e_k\}$ $(k=0,1,\dots,n)$ be the set of first $k$ items in $I$,
$\OPT\in \argmax\{s(X)\mid s(X)\le 1,~X\subseteq I_n\}$ be an optimal solution, and 
$B^* \in \argmax\{s(B') \mid B' \subseteq B_n,\  s(B') \le 1\}$ be an outcome of Algorithm~\ref{alg:17-9sqrt(3)}. 

We see that the algorithm does not violate the buffer constraint.
\begin{observation}
  Algorithm~\ref{alg:17-9sqrt(3)} is feasible, i.e., $s(B_i)\le R$ for all $i\in\{0,1,\dots,n\}$.
\end{observation}
\begin{proof}
If the condition at Line~\ref{line:17-9sqrt(3)_win} is satisfied in round $i$, then we have $s(B_i)\le 1\le R$.
Thus, we assume that the condition is not satisfied in round $i$.

The total size of medium items is not larger than $R$ because the total size of an $M_3$-item and an $M_4$-item is at most $r^2+r\le R$.
In addition, all the small items must fit the buffer since otherwise there exists a subset $B'\subseteq B_i$ that satisfies the condition at Line~\ref{line:17-9sqrt(3)_win}.
This is because there exists a partition of $B_i\cap M=B_i^{(1)}\cup B_i^{(2)}$ such that $s(B_i^{(1)})\le r$ and $s(B_i^{(2)})\le R-r$.
\end{proof}

Next, we observe that the algorithm outputs a solution $B^*$ such that $r\le s(B^*)\le 1$ if 
(i) $|I_n\cap L|\ge 1$ or (ii) $|I_n\cap M_1| \ge 1$ and $|I_n\cap (M_2\cup M_3)| \ge 1$.
Note that, if $s(B^*)\geq r$, then the competitive ratio is at most $\frac{1}{r}$.
\begin{observation}\label{17-9sqrt(3):L}
If there exists a large item in $I$, then $r\le s(B^*)\le 1$.
\end{observation}
\begin{proof}
This statement holds by $r\le s(e)\le 1$ for all $e\in L$.
\end{proof}

\begin{observation}\label{17-9sqrt(3):M_2_M_3}
  If $|I_n\cap M_1| \ge 1$ and $|I_n\cap (M_2\cup M_3)| \ge 1$, then $r\le s(B^*)\le 1$.
\end{observation}
\begin{proof}
  Let $k$ be the smallest round such that $|I_k\cap M_1| \ge 1$ and $|I_k\cap (M_2\cup M_3)| \ge 1$.
  Then, $B_{k-1}\cup\{e_k\}$ contains an $M_1$-item $e'$ and an $M_2$- or $M_3$-item $e''$ because the algorithm always keeps the smallest medium item (if $s(B_{k-1})\notin [r,1]$).
  Thus, the claim holds since $r=(1-r)+(2r-1)\le s(e')+s(e'')\le (2r-1)+r^2\le 1$.
\end{proof}

By Observation~\ref{17-9sqrt(3):L}, it is sufficient to consider the case that $I_n\cap L=\emptyset$.
Let $\OPT_M\in \argmax\{s(X) \mid X \subseteq I_n\cap M, s(X)\le 1\}$.
Then, the competitive ratio of Algorithm~\ref{alg:17-9sqrt(3)} is
$$\frac{s(\OPT)}{s(B^*)}
\le \frac{s(\OPT_M) + s(I_n\cap S)}{s(B^*\cap M) + s(I_n\cap S)}
\le \frac{s(\OPT_M)}{s(B^*\cap M)}.$$
Therefore, we can assume $I_n\subseteq M$ without loss of generality.

Now, we prove Theorem~\ref{thm:17-9sqrt(3)}.
We consider the following three cases separately:
$|\OPT|\ge 3$, $|\OPT|=2$, and $|\OPT|\le 1$.
\begin{lemma}\label{lem:17-9sqrt(3):3}
If $I_n\subseteq M$ and $|\OPT|\ge 3$, then the competitive ratio is at most $\frac{\sqrt{16R+1}-1}{2R} (=1/r)$.
\end{lemma}
\begin{proof}
  We have $|\OPT \cap M_1| = 3$ since otherwise $s(\OPT)>1$. 
  Let $a,b,c \in \OPT \cap M_1$ arrive in the order of $a,b,c$. 
  Let $k$ be the round such that $c$ arrives.
  Suppose to the contrary that $s(B_k)\notin [r,1]$.
  In the round $k$, the algorithm keeps two $M_1$-items whose total size is at most $s(a) + s(b)$, or has an $M_4$-item with size at most $s(a) + s(b)$. 
  In both cases, there exists $B'\subseteq B_{k-1}\cup\{e_k\}$ such that $r\le 3(1-r)\le s(B')\le s(a)+s(b)+s(c)\le 1$, a contradiction.
  Hence, $r\le s(B^*)\le 1$ and the claim holds.
\end{proof}

\begin{lemma}\label{lem:17-9sqrt(3):1}
If $I_n\subseteq M$ and $|\OPT|\le 1$, then the competitive ratio is at most $\frac{\sqrt{16R+1}-1}{2R} (=1/r)$.
\end{lemma}
\begin{proof}
Since the claim is clear when $\OPT=\emptyset$, we assume $|\OPT|= 1$.
Let $\OPT=\{m^*\}$.
Note that $m^*$ is the largest item in $I_n$ and $s(B_n)\le s(m^*)<r$.

If $m^*\in M_1\cup M_2$, we have $I\coloneqq(m^*)$ (otherwise, there exists a better solution than $\OPT$).
Thus, $B^*=\{m^*\}$ and the competitive ratio is one.

If $m^*\in M_3$, we have $I_n\cap M_1=\emptyset$ by Observation~\ref{17-9sqrt(3):M_2_M_3}. 
Let $e^* \in \argmin\{s(e) \mid e\in I_n\}$. 
Then, we have $s(e) \geq 1-r^2$ (otherwise $\{m^*,e^*\}$ is a better solution).
Thus, $s(B^*)\ge 1-r^2$ and the competitive ratio is at most $r^2/(1-r^2) \le 1/r$.

If $m^*\in M_4$, we have $s(B^*)\ge r^2$ and the competitive ratio is at most $r/r^2 = 1/r$.
\end{proof}

\begin{lemma}\label{lem:17-9sqrt(3):2}
  If $I_n\subseteq M$ and $|\OPT|=2$, then the competitive ratio is at most $\frac{\sqrt{16R+1}-1}{2R} (=1/r)$.
\end{lemma}
\begin{proof}
Let $\OPT=\{e_{j},e_{k}\}$ with $j < k$.
If $r\le s(B_i)\le 1$ for some $i$, then $r\le s(B^*)\le 1$ and the competitive ratio is at most $\frac{\sqrt{16R+1}-1}{2R}$.
Thus, suppose that $s(B_i)\notin [r,1]$ for all $i$.

\begin{description}
\item[Case 1:] Suppose that $e_{j} \in M_1$.
  By Observation~\ref{17-9sqrt(3):M_2_M_3}, we can assume $I_n\cap (M_2\cup M_3)=\emptyset$.
  Let $e^* \in \argmin\{s(e) \mid e \in B_{k-1}\}$.
  If $e_k\in M_4$, then we have $r\le s(B_k)\le 1$ by $r \leq s(e^*) + s(e_k) \leq s(e_j) + s(e_k) \leq 1$, a contradiction.
  Otherwise, i.e., $e_k\in M_1$, we consider the following two cases.

\item[Case 1-1:] Suppose that $e_j,e_k\in M_1$ and $B_n \cap M_4 \ne \emptyset$.
  Let $m_1 \in \argmin\{s(e)\mid e \in B_n\cap M_1\}$ (i.e., $s(m_1)=\min_{e\in I_n} s(e)$) and $m_2 \in B_n\cap M_4$.
  By $s(m_1)+s(m_2)\ge (1-r)+r^2\ge r$ and $s(B_n)\notin [r,1]$, we have $s(m_1)+s(m_2)> 1$.
  If $s(m_2)\ge R-r$, then the competitive ratio is at most $\frac{2(2r-1)}{R-r} = 1/r$.
  Thus, we have $R-r>s(m_2)$ and $s(m_1) > 1-s(m_2) > 1-R+r$.
  If $s(m_1) \geq r/2$, there exists a round $\ell$ such that $r\le s(B_\ell)\le 1$.
  Thus, we can assume $s(m_1) < r/2$ and hence we have $s(m_2) > 1-s(m_1) > 1 - r/2$.
  In addition, we have $s(m)<R-1$ for any $m\in I_n \cap M_1$ by $s(m)+s(m_1)<r$.
  Therefore, the competitive ratio is at most $\frac{2(R-1)}{1-r/2} < 1/r$.

\item[Case 1-2:] Suppose that $e_j,e_k\in M_1$ and $B_n \cap M_4 = \emptyset$.
  In this case, we have $|B_n\cap M_1| \ge 2$.
  Let $m_1 \in \argmin\{s(e)\mid e \in B_n\cap M_1\}$ and $m_2 \in \argmin\{s(e)\mid e \in (B_n\cap M_1)\setminus\{m_1\}\}$.
  Note that $s(m_1)=\min_{e\in I_n}s(e)$ and $s(m_2)=\min_{e\in I_n\setminus\{m_1\}}s(e)$.
  We can assume $s(m_1) + s(m_2) < R-r $ since otherwise the competitive ratio is at most $\frac{2(2r-1)}{R-r}=1/r$.
  If $|I_n\cap M_1| = 2$, we have $\{m_1,m_2\}=\{e_j,e_k\}$ and the competitive ratio is one.
  Hence, we assume $|I_n\cap M_1| \ge 3$.
  Let $m_3 \in \argmin\{s(e)\mid e \in I_n \setminus \{ m_1,m_2\}\}$.
  If $|\argmin\{s(e)\mid e \in I_n \setminus \{ m_1,m_2\}\}|>1$,
  we pick the one that maximizes $\max\{i\mid m_3\in B_{i-1}\cup\{e_i\}\}$.

  Suppose that $m_3\in B_n$.
  Then, we can assume  $s(m_2) + s(m_3) < R-r$ since otherwise the competitive ratio is at most $\frac{2(2r-1)}{R-r}=1/r$.
  We have $s(m_1) > 1-R+r$ (otherwise $r\le 3(1-r)\le s(m_1)+s(m_2)+s(m_3)\le 1$) and hence $R-1 > s(m) > 1-R+r$ for any $m\in I_n \cap M_1$.
  Thus, the competitive ratio is at most $\frac{2(R-1)}{2(1-R+r)} < 1/r$.

  Suppose that $m_3\notin B_n$.
  Let $\ell$ be the round such that $m_3$ is discarded (i.e., $m_3\in B_{\ell-1}\cup\{e_{\ell}\}$ and $m_3\notin B_{\ell}$).
  Here, we have two cases to consider: $B_{\ell}=\{m_1,m_2\}$ and $B_{\ell}=\{m',m_4\}$ with $m'\in\{m_1,m_2\}$ and $m_4\in M_4$.
  If $B_{\ell}=\{m_1,m_2\}$, then $s(m_1)+s(m_2)>R-r$ (otherwise $B_{\ell}=\{m_1,m_2,m_3\}$).
  If $B_{\ell}=\{m',m_4\}$ with $m'\in\{m_1,m_2\}$ and $m_4\in M_4$,
  we have $s(m_1)+s(m_2)>R-r$ (otherwise $s(m_1)+s(m_2)\le R-r$ implies $|B_n|\ge 3$, a contradiction).
  Hence, in both cases, the competitive ratio is at most $\frac{2(2r-1)}{R-r}=\frac{1}{r}$.
  
\item[Case 2:] Suppose that $e_j \in M_2\cup M_3$.
  By Observation~\ref{17-9sqrt(3):M_2_M_3}, we can assume $I_n\cap M_1=\emptyset$.
  Let $e^* \in \argmin\{s(e')\mid e' \in B_{k -1}\}$. Then, $e^*\in M_2\cup M_3$.
  By $r\le 2(2r-1)\leq s(e^*) + s(e_k) \le s(e_j)+s(e_k)\le 1$, we have $r\le s(B_k)\le 1$, a contradiction.

\item[Case 3:] Suppose that $e_j \in M_4$.
  In the round $k$, the algorithm keeps an $M_4$-item whose size is at most $s(e_j)$ or two $M_1$-items whose total size is at most $s(e_j)$.
  Thus, $r\le s(B_k)\le 1$, a contradiction. \qedhere
\end{description}
\end{proof}

Hence, Theorem~\ref{thm:17-9sqrt(3)} is proved from Lemmas~\ref{lem:17-9sqrt(3):3}, \ref{lem:17-9sqrt(3):1}, and \ref{lem:17-9sqrt(3):2}.

\subsection{\texorpdfstring{$2\sqrt{3}-2\le R\le \frac{3}{2}$}{2sqrt{3}-2le Rle frac{3}{2}}}
In this subsection, we consider the problem for $2\sqrt{3}-2< R\le \frac{3}{2}$.
We prove that the competitive ratio is $R/2$ for the case.

\subsubsection{Lower bound}
First, we give a lower bound of the competitive ratio.

\begin{theorem}\label{thm:prop/removable_lower_iii}
  For any sufficiently small positive real $\epsilon$, 
  the competitive ratio of the proportional\&removable online knapsack problem with a buffer
  is at least $2/R-\epsilon$ when $2\sqrt{3}-2<R<2$.
\end{theorem}
\begin{proof}
  Consider the following input sequence:
  \begin{align*}
  \frac{R}{2},\frac{R^2+\epsilon}{4},R-1.
  \end{align*}
  Here, we have $R/2>R^2/4>R-1$ and $R^2/4+(R-1)>1$.
  
  If an online algorithm discards the first item,
  then the competitive ratio is at least $\frac{R/2}{(R^2+\epsilon)/4}>\frac{2/R}{1+\epsilon/R^2}>2/R-\epsilon$.
  If an online algorithm discards the third item,
  suppose that the next item has size $2-R$, and 
  then the competitive ratio is at least $\frac{(R-1)+(2-R)}{R/2}\ge 2/R$.

  If an online algorithm discards the second item,
  then we assume that the fourth item with size $1-R/2+\epsilon/4$ arrives.
  Now, it must discard an item because $R/2+(R-1)+(1-R/2+\epsilon/4)>R$.
  If it removes the first item, suppose that the next item has size $1-R/2$,
  and then the competitive ratio is at least $\frac{R/2+(1-R/2)}{(R-1)+(1-R/2+\epsilon/4)}>2/R-\epsilon$.
  If it removes the third item, suppose that the next item has size $1-(R^2+\epsilon)/4$,
  and then the competitive ratio is at least $\frac{(R^2+\epsilon)/4+(1-(R^2+\epsilon)/4)}{R/2}>2/R-\epsilon$.
  Finally, if it discards the item with size $1-R/2+\epsilon/4$, 
  suppose that the next item has size $R/2-\epsilon/4$,
  and then the competitive ratio is at least $\frac{(1-R/2+\epsilon/4)+(R/2-\epsilon/4)}{R/2}=2/R$.
  
  Hence, the competitive ratio is at least $2/R-\epsilon$.
\end{proof}

\subsubsection{Upper bound}
In this subsubsection, an item \(e\) is called {\em small}, {\em medium}, and {\em large}
if \(s(e)\le 1-R/2\), \(1-R/2< s(e)<R/2\), and \(R/2\le s(e)\), respectively.
Let $S$, $M$, and $L$ respectively denote the sets of small, medium, and large items.
$M$ is further partitioned into four subsets $M_i$ for $ 1 \le i \le 4$,
where $M_1$, $M_2$, $M_3$, and $M_4$ respectively denote the set of the items $e$ with size
\(1-R/2< s(e)< R/4\), \(R/4\le s(e)< 1/2\), \(1/2\le s(e)<R^2/4\), and \(R^2/4\le s(e)<R/2\)
(see Figure~\ref{fig:sizes2/R}).
An item $e$ is also called an {\em $M_i$-item} if \(e\in M_i\).

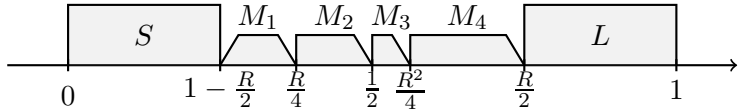
\begin{figure}[htbp]
\begin{center}
  \begin{tikzpicture}[scale=0.8]
    \draw[thick,->] (-1,0)--(11,0);
    \filldraw[thick,fill=gray,fill opacity=0.1] (0,0)--(0,1)--(2.5,1)--(2.5,0);
    \filldraw[thick,fill=gray,fill opacity=0.1] (2.5,0)--(2.8,.5)--(3.45,.5)--(3.75,0);
    \filldraw[thick,fill=gray,fill opacity=0.1] (3.75,0)--(3.75,.5)--(4.7,.5)--(5,0);
    \filldraw[thick,fill=gray,fill opacity=0.1] (5,0)--(5,.5)--(5.325,.5)--(5.625,0);
    \filldraw[thick,fill=gray,fill opacity=0.1] (5.625,0)--(5.625,.5)--(7.2,.5)--(7.5,0);
    \filldraw[thick,fill=gray,fill opacity=0.1] (7.5,0)--(7.5,1)--(10,1)--(10,0);
    \draw (1.25,0.5) node {$S$};
    \draw (3.125,0.8) node {$M_1$};
    \draw (4.375,0.8) node {$M_2$};
    \draw (5.3125,0.8) node {$M_3$};
    \draw (6.5625,0.8) node {$M_4$};
    \draw (8.75,0.5) node {$L$};
    \draw[thick] (0.00,-0.1)--(0.00,0.1);
    \draw (0,-0.5) node {$0$};
    \draw[thick] (2.5,-0.1)--(2.5,0.1);
    \draw (2.5,-0.4) node {$1-\frac{R}{2}$};
    \draw[thick] (3.75,-0.1)--(3.75,0.1);
    \draw (3.75,-0.4) node {$\frac{R}{4}$};
    \draw[thick] (5,-0.1)--(5,0.1);
    \draw (5,-0.4) node {$\frac{1}{2}$};
    \draw[thick] (5.625,-0.1)--(5.625,0.1);
    \draw (5.625,-0.4) node {$\frac{R^2}{4}$};
    \draw[thick] (7.5,-0.1)--(7.5,0.1);
    \draw (7.5,-0.4) node {$\frac{R}{2}$};
    \draw[thick] (10.00,-0.1)--(10.00,0.1);
    \draw (10.0,-0.4) node {$1$};
  \end{tikzpicture}
  \caption{Item partition for $2\sqrt{3}-2\le R\le \frac{3}{2}$.}\label{fig:sizes2/R}
\end{center}
\end{figure}

We consider Algorithm~\ref{alg:2/R}.
If the algorithm can select a set of items $B'$ such that $r\le s(B')\le 1$, it keeps the set $B'$ until the end.
Otherwise, it first selects the smallest $M_1$- and $M_2$-items.
If there exists no $M_2$-item or the total size of the two smallest $M_1$ items is not larger than the smallest $M_4$-item,
then it picks the second smallest $M_1$-item.
If there exists no item in $M_1\cup M_2$ or $M_4$, then it picks the smallest $M_3$-item.
If the algorithm does not take three items in $M_1\cup M_2$, then it picks the smallest $M_4$-item.
All the small items are taken into the remaining space.
We show that the algorithm is optimal when $2\sqrt{3}-2\le R\le \frac{3}{2}$.

\begin{algorithm}[htb]
  \caption{$\frac{2}{R}$-competitive algorithm}\label{alg:2/R}
  $B_0\ot\emptyset$\;
  \For{\(i\ot 1,2,\dots\)}{%
    \lIf{$\exists B'\subseteq B_{i-1}\cup\{e_i\}$ such that $R/2\le s(B')\le 1$}{$B_i\ot B'$}\label{line:2/R_win}
    \Else{
      let $B_i\ot \emptyset$, $A_j\ot (B_{i-1}\cup\{e_i\})\cap M_j~(j=1,2,3,4)$\;
      \If{$|A_1|>0$}{
        let $a\in\argmin_{e\in A_1}s(e)$ and $B_i\ot B_i\cup\{a\}$\;
        \If{$|A_1|\ge 2$}{
          let $b\in\argmin_{e\in A_1\setminus\{a\}}s(e)$\;
          \lIf{$|A_2|=0$ or $s(a)+s(b)\le\min_{e\in A_4}s(e)$}{$B_i\ot B_i\cup\{b\}$}
        }
      }
      \If{$|A_2|>0$}{
        let $a\in\argmin_{e\in A_2}s(e)$ and $B_i\ot B_i\cup\{a\}$\;
      }
      \If{$|A_3|>0$ and ``$|A_1\cup A_2|=0$ or $|A_4|=0$''}{
        let $a\in\argmin_{e\in A_3}s(e)$ and $B_i\ot B_i\cup\{a\}$\;
      }
      \If{$|A_4|>0$ and $|B_i\cap (M_1\cup M_2)|\le 2$}{
        let $a\in\argmin_{e\in A_4}s(e)$ and $B_i\ot B_i\cup\{a\}$\;
      }
      $B_i\ot B_i\cup ((B_{i-1}\cup\{e_i\})\cap S)$\;
    }
  }
\end{algorithm}

\begin{theorem}\label{thm:2/R}
  Algorithm~\ref{alg:2/R} is $\frac{2}{R}$-competitive when $2\sqrt{3} - 2\le R \le \frac{3}{2}$.
\end{theorem}

Let $I = (e_1,\dots,e_n)$ be the input sequence, 
$I_k=\{e_1,\dots,e_k\}$ $(k=0,1,\dots,n)$ be the set of first $k$ items in $I$,
$\OPT\in \argmax\{s(X)\mid s(X)\le 1,~X\subseteq I_n\}$ be an optimal solution, and 
$B^* \in \argmax\{s(B') \mid B'\subseteq B_n,\  s(B') \le 1\}$ be an outcome of Algorithm~\ref{alg:2/R}. 

We see that the algorithm does not violate the buffer constraint.
\begin{observation}
  Algorithm~\ref{alg:2/R} is feasible, i.e., $s(B_i)\le R$ for all $i\in\{0,1,\dots,n\}$.
\end{observation}
\begin{proof}
If the condition at Line~\ref{line:2/R_win} is satisfied in round $i$, then we have $s(B_i)\le 1\le R$.
Thus, we assume that the condition is not satisfied in round $i$.
Here, we remark that $s(B_i\cap (M_1\cup M_2))<R/2$ (otherwise the condition is satisfied).
If $B_i\cap M_3\ne\emptyset$, then the total size of medium items in $B_i$ is at most $R^2/4+R/2<R$.
If $B_i\cap M_3=\emptyset$, then the total size of medium items in $B_i$ is at most $R/2+R/2\le R$.
In addition, all the small items must fit the buffer since otherwise there exists a subset $B'\subseteq B_i$ that satisfies the condition at Line~\ref{line:2/R_win}.
\end{proof}

Next, we observe some sufficient conditions such that the algorithm outputs a solution $B^*$ such that $R/2\le s(B^*)\le 1$.
Note that, if $s(B^*)\geq R/2$, then the competitive ratio is at most $\frac{2}{R}$.
\begin{observation}\label{obs:2/R:L}
If $|I_n\cap L|\ge 1$, then $R/2\le s(B^*)\le 1$.
\end{observation}
\begin{proof}
This statement holds by $R/2\le s(e)\le 1$ for all $e\in L$.
\end{proof}

\begin{observation}\label{obs:2/R:M2M2}
If $|I_n\cap M_2|\ge 2$, then $R/2\le s(B^*)\le 1$.
\end{observation}
\begin{proof}
  Let $k$ be the smallest round such that $|I_k\cap M_2| \ge 2$.
  Then, $e_k\in M_2$ and $B_{k-1}$ contains an $M_2$-item $e'$  (if $B_{k-1}\notin [r,1]$).
  Thus, the claim holds since $R/2=R/4+R/4\le s(e_k)+s(e')\le 1$.
\end{proof}

\begin{observation}\label{obs:2/R:M1M3}
If $|(B_{i-1}\cup\{e_i\})\cap M_1|\ge 1$ and $|(B_{i-1}\cup\{e_i\})\cap M_3|\ge 1$ for some $i$, 
then $R/2\le s(B^*)\le 1$.
\end{observation}
\begin{proof}
Let $e'\in (B_{i-1}\cup\{e_i\})\cap M_1$ and $e''\in (B_{i-1}\cup\{e_i\})\cap M_3$.
The claim holds since $R/2\le (1-R/2)+1/2\le s(e')+s(e'')\le 1/2+R^2/4\le 1$.
\end{proof}

By Observation~\ref{obs:2/R:L}, it is sufficient to consider the case that $I_n\cap L=\emptyset$.
Let $\OPT_M\in \argmax\{s(X) \mid X \subseteq I_n\cap M, s(X)\le 1\}$.
Then, the competitive ratio of Algorithm~\ref{alg:2/R} is
$$\frac{s(\OPT)}{s(B^*)}
\le \frac{s(\OPT_M) + s(I_n\cap S)}{s(B^*\cap M) + s(I_n\cap S)}
\le \frac{s(\OPT_M)}{s(B^*\cap M)}.$$
Therefore, we can assume $I_n\subseteq M$ without loss of generality.

Now, we prove Theorem~\ref{thm:2/R}.
As $s(e)>1/4$ for all $e\in M$,
we consider the following three cases separately: $|\OPT|=3$, $|\OPT|=2$, and $|\OPT|\le 1$.
\begin{lemma}\label{lem:2/R:1}
If $I_n\subseteq M$ and $|\OPT|=3$, then the competitive ratio is at most $2/R$.
\end{lemma}
\begin{proof}
  We have two cases to consider: (i) $|\OPT \cap M_1| = 3$ and (ii) $|\OPT\cap M_1|=2$ and $|\OPT\cap M_2|=1$.
  Without loss of generality, we assume that $s(B_n)\notin [R/2,1]$ (otherwise the competitive ratio is clearly at most $2/R$).

  Suppose that $|\OPT\cap M_1|=3$.
  Let $\{a,b,c\}=\OPT \cap M_1$ arrive in the order of $a,b,c$.
  Let $j$ be the round such that $c$ arrives.
  At the beginning of round $j$, the algorithm keeps two $M_1$-items whose total size is at most $s(a) + s(b)$, or keeps an $M_4$-item with size at most $s(a) + s(b)$.
  Thus, we have $R/2\le s(B_j)\le 1$ since $3(1-R/2)\ge R/2$ and $s(a)+s(b)+s(c)\le 1$, a contradiction.

  Suppose that $|\OPT\cap M_1|=2$.
  Let $j$ be the round such that $j=\max\{i\mid  e_i\in\OPT\}$.
  Then, $B_{j-1}\cup\{e_j\}$ contains an $M_2$-item with size at most $s(\OPT\cap M_2)$.
  In addition, $B_{j-1}\cup\{e_j\}$ contains two $M_1$ items whose total size is at most $s(\OPT\cap M_1)$, or keeps an $M_4$-item with size at most $s(\OPT\cap M_1)$.
  Thus, we have $R/2\le s(B_j)\le 1$ since $2(1-R/2)+R/4\ge R/2$ and $s(\OPT)=s(\OPT\cap M_1)+s(\OPT\cap M_2)\le 1$, a contradiction.
\end{proof}

\begin{lemma}\label{lem:2/R:2}
If $I_n\subseteq M$ and $|\OPT|\le 1$, then the competitive ratio is at most $2/R$.
\end{lemma}
\begin{proof}
  Since the claim is clear when $\OPT=\emptyset$, we assume $|\OPT|= 1$.
  Let $\OPT=\{m^*\}$.
  Note that $m^*$ is the largest item in $I_n$ and $s(B_n)\le s(m^*)<R/2$.

  If $m^*\in M_1\cup M_2$, then $I\coloneqq(m^*)$ and the competitive ratio is one.
  If $m^*\in M_3$, then the competitive ratio is at most $\frac{R^2/4}{1/2}\le 2/R$ since $B_n\cap M_3\ne\emptyset$ by $I_n\cap M_4=\emptyset$.
  If $m^*\in M_4$, then the competitive ratio is at most $\frac{R/2}{\min\{R^2/4, (1-R/2)+R/4\}}=\frac{R/2}{R^2/4}=2/R$ by $R^2/4\le (1-R/2)+R/4$
  since $B_n\cap M_4\ne\emptyset$ or ``$B_n\cap M_1\ne\emptyset$ and $B_n\cap M_2\ne\emptyset$''.
\end{proof}

\begin{lemma}\label{lem:2/R:3}
  If $I_n\subseteq M$ and $|\OPT|=2$, then the competitive ratio is at most $2/R$.
\end{lemma}
\begin{proof}
  Without loss of generality, we assume that $s(B_n)\notin [R/2,1]$ (otherwise the competitive ratio is clearly at most $2/R$).
  Let $\OPT=\{e_j,e_k\}$ with $s(e_j)\le s(e_k)$.
  We have six cases to consider:
  (i) $e_j,e_k\in M_1$,
  (ii) $e_j\in M_1$, $e_k\in M_2$,
  (iii) $e_j\in M_1$, $e_k\in M_3$,
  (iv) $e_j\in M_1$, $e_k\in M_4$,
  (v) $e_j\in M_2$, $e_k\in M_2$, and
  (vi) $e_j\in M_2$, $e_k\in M_3\cup M_4$.
  \begin{description}
  \item[Case (i):] Suppose that $e_j,e_k\in M_1$. 
    In this case $I_n\cap M_2=\emptyset$ by the optimality of $\{e_j,e_k\}$.
    Note that $e_j,e_k$ are largest two items of $I_n\cap M_1$ and $s(\OPT)< R/4+R/4=R/2$.
    If $|I_n\cap M_1|=2$, then the competitive ratio is clearly one.
    If $|I_n\cap M_1|\ge 3$, then $s(B^*) > 1-s(e_k)$ (because any triplet of items in $I_n\cap M_1$ cannot be packed together into the knapsack) and the competitive ratio is at most
    \begin{align*}
      \frac{s(e_j)+s(e_k)}{s(B^*)}
      < \frac{2s(e_k)}{1-s(e_k)}
      < \frac{R/2}{1-R/4}
      =\frac{2}{R}\cdot \frac{R^2}{4-R}<\frac{2}{R},
    \end{align*}
    where the last inequality holds since $f(R)\coloneqq\frac{R^2}{4-R}$ is monotone increasing for $R\in[0,4)$
    and $f(3/2)=9/10<1$.

  \item[Case (ii):] Suppose that $e_j\in M_1$ and $e_k\in M_2$.
    By Observation~\ref{obs:2/R:M2M2}, $|I_n\cap M_2| = 1$ and $e_k \in B_n$.
    Thus, the competitive ratio is at most
    \begin{align*}
    \frac{s(e_j)+s(e_k)}{(1-R/2)+s(e_k)}
    \le\frac{R/4+s(e_k)}{(1-R/2)+s(e_k)}
    \le \frac{R/4+R/4}{(1-R/2)+R/4}=\frac{2}{R}\cdot \frac{R^2}{4-R}<\frac{2}{R}.
    \end{align*}

  \item[Case (iii):] Suppose that $e_j\in M_1$ and $e_k\in M_3$.
    If $B_n\cap M_3\ne\emptyset$ then $R/2\le s(B_n)\le 1$ by Observation~\ref{obs:2/R:M1M3}.
    Otherwise, let $\ell$ be the round such that the item $m^*\in \argmin\{s(e)\mid e\in I_n\cap M_3\}$ is discarded,
    i.e., $m^*\in B_{\ell-1}\cup\{e_{\ell}\}$ and $m^*\notin B_{\ell}$.
    Note that, by Observation~\ref{obs:2/R:M1M3}, $I_\ell \cap M_1 = \emptyset$.
    Then, $|B_\ell\cap M_2|=|B_\ell\cap M_4|=1$ and $|B_\ell\cap M_1|=0$.
    Let $B_\ell\cap M_2=\{m_1\}$ and $B_\ell\cap M_4=\{m_2\}$.
    We remark that $I_n\cap M_2=\{m_1\}$ by Observation~\ref{obs:2/R:M2M2}.
    Here, $j>\ell$ since otherwise $|B_\ell\cap M_1|\ge 1$.
    Also, we have $s(m_1)>1-s(m^*)>1-R^2/4$ and $s(e_j)<R/2-s(m_1)<R/2-(1-R^2/4)$.
    Hence, the competitive ratio is at most
    \begin{align*}
    \frac{s(e_j)+s(e_k)}{s(m_1)+(1-R/2)}
    \le\frac{(R/2-1+R^2/4)+(R^2/4)}{(1-R^2/4)+(1-R/2)}
    =\frac{2}{R}\cdot \frac{R(R+2)(R-1)}{(2-R)(R+4)}
    <\frac{2}{R},
    \end{align*}
    where the last inequality holds since $g(R)\coloneqq\frac{R(R+2)(R-1)}{(2-R)(R+4)}$ is monotone increasing for $R\in[0,2)$
    and $g(3/2)=21/22<1$.

  \item[Case (iv):] Suppose that $e_j\in M_1$ and $e_k\in M_4$.
    If $B_n\cap M_4\ne\emptyset$ then $R/2\le s(B_n)\le 1$
    since $B_n$ contains the smallest $M_1$-item $e'$ and the smallest $M_4$-item $e''$,
    which satisfy $R/2\le s(e')+s(e'')\le s(e_j)+s(e_k)\le 1$.
    Otherwise, let $\ell$ be the round such that the item $m^*\in \argmin\{s(e)\mid e\in I_n\cap M_4\}$ is discarded,
    i.e., $m^*\in B_{\ell-1}\cup\{e_{\ell}\}$ and $m^*\notin B_{\ell}$.
    Then $|B_\ell\cap M_1|=2$ and $|B_\ell\cap M_2|=1$.
    Let $|B_\ell\cap M_1|=\{m_1,m_2\}$ with $s(m_1)\le s(m_2)$.
    Then, we have $s(m_1)+s(m_2)\le s(m^*)\le s(e_k)$ and $s(m_1)+s(m^*)>1$.
    Note that $j>\ell$ since otherwise $\{e_j,m^*\}\subseteq B_{j-1}\cup\{e_{\ell}\}$ and $R/2\le (1-R/2)+R^2/4\le s(e_j)+s(m^*)\le s(e_j)+s(e_k)\le 1$, a contradiction.
    Then, there exist two items $m_1',m_2'\in B_{j-1}\cap M_1$ such that $s(m_1')+s(m_2')\le s(m_1)+s(m_2)$,
    and hence $B'\coloneqq \{e_j,m_1',m_2'\}\subseteq B_{j-1}\cup\{e_j\}$ satisfies $R/2\le 3(1-R/2)\le s(B')\le s(e_j)+s(m_1)+s(m_2)\le s(e_j)+s(e_k)\le 1$, a contradiction.
    
  \item[Case (v):] Suppose that $e_j,e_k\in M_2$. In this case, we have $R/2\le s(B^*)\le 1$ by Observation~\ref{obs:2/R:M2M2}.
    
  \item[Case (vi):] Suppose that $e_j\in M_2$ and $e_k\in M_3\cup M_4$.
    If $B_n\cap (M_3\cup M_4) \ne\emptyset$, then $R/2\le s(B_n)\le 1$
    since $B_n$ contains $e'\in\argmin\{s(e)\mid e\in I_n\cap M_2\}$ and $e''\in\argmin\{s(e)\mid e\in I_n\cap (M_3\cup M_4)\}$,
    which satisfy $R/2\le R/4+1/2\le s(e')+s(e'')\le s(e_j)+s(e_k)\le 1$.
    Otherwise, let $\ell$ be the round such that the item $m^*\in \argmin\{s(e)\mid e\in I_n\cap (M_3\cup M_4)\}$ is discarded,
    i.e., $m^*\in B_{\ell-1}\cup\{e_{\ell}\}$ and $m^*\notin B_{\ell}$.
    Then, we have $B_\ell\cap M_2=\{e_j\}$ by Observation~\ref{obs:2/R:M2M2}
    and $R/2\le s(e_j)+s(m^*)\le s(e_j)+s(e_k)\le 1$, a contradiction. \qedhere
  \end{description}
\end{proof}

Hence, Theorem~\ref{thm:2/R} is proved from Lemmas~\ref{lem:2/R:1}, \ref{lem:2/R:2}, and \ref{lem:2/R:3}.

\subsection{General \texorpdfstring{$R$}{R}}

In this subsection, we consider proportional\&removable case with general $R$.
By Theorem~\ref{thm:general_removable_general}, the upper bound of the competitive ratio is $1+O(\log R/R)$.
Hence, we only give a lower bound of the competitive ratio.

\begin{theorem}\label{thm:lower_proportional_removable_general}
  For any positive real $\epsilon< 1$, 
  the competitive ratio of the proportional\&removable online knapsack problem with a buffer is
  at least $1 + \frac{1}{\lceil 2R\rceil+1} - \epsilon$.
\end{theorem}

\begin{proof}
Let $n\coloneqq \lceil 2R\rceil+1$ and let $\ALG$ be an online algorithm.
Consider the item sequence $I\coloneqq(e_1,\dots,e_{n-1},e_n)$ where
$s(e_i)=\frac{i}{n}+\frac{\epsilon}{n^2}$ for $i=1,\dots,n-1$ (we will set $s(e_n)$ later depending on $\ALG$).
At the end of $(n-1)$st round, $\ALG$ must discard at least one item because $\sum^{n-1}_{i=1} s(e_i) > \frac{n-1}{2}=\frac{\lceil 2R\rceil}{2} \ge R$.
Suppose that $\ALG$ discards $e_j$, and let $s(e_n)=1-s(e_j)$.
Then, we have $\OPT(I)=s(e_j)+s(e_n)=1$.
We will prove that $\ALG(I)$ is at most $1-(1-\epsilon)/n$,
which implies that the competitive ratio of $\ALG$ is at least $\frac{1}{1-(1-\epsilon)/n}\ge 1+(1-\epsilon)/n\ge 1+\frac{1}{\lceil 2R\rceil+1}-\epsilon$.

Let $B^*$ be the output of $\ALG$, i.e., $s(B^*)=\ALG(I)$.
We have two cases to consider: $e_n\notin B^*$ and $e_n\in B^*$.
\begin{description}
\item[Case 1:] If $e_n\not\in B^*$, then we have $s(B^*)=\frac{\sum_{e_i\in B^*} i}{n}+\frac{|B^*|\cdot\epsilon}{n^2}$.
We assume $B^*\ne\emptyset$ since otherwise $s(B^*)=0$.
Since $s(B^*)\le 1$ and $\frac{|B^*|\cdot\epsilon}{n^2}>0$, we have $\frac{\sum\{i\mid e_i\in B^*\}}{n}\le\frac{n-1}{n}$.
Hence, we obtain $s(B^*)\le \frac{n-1}{n}+\frac{n\cdot\epsilon}{n^2}=1-\frac{1-\epsilon}{n}$.

\item[Case 2:] If $e_n\in B^*$, then we have $s(B^*)=\frac{(n-j)+\sum\{i\mid e_i\in B^*\setminus\{e_n\}\}}{n}+\frac{(|B^*|-2)\cdot\epsilon}{n^2}$.
We assume $|B^*|\ge 3$ since otherwise $s(B^*)\le \frac{n-1}{n}$ by $e_j\not\in B^*$.
Since $s(B^*)\le 1$ and $\frac{(|B^*|-2)\cdot\epsilon}{n^2}>0$, we have $\frac{(n-j)+\sum\{i\mid e_i\in B^*\setminus\{e_n\}\}}{n}\le\frac{n-1}{n}$.
Hence, we obtain $s(B^*)\le \frac{n-1}{n}+\frac{n\cdot\epsilon}{n^2}=1-\frac{1-\epsilon}{n}$. \qedhere
\end{description}
\end{proof}

\section*{Acknowledgments}
  The first author was supported by RGC (HKU716412E) and NSFC (11571060).
  The second author was supported by JSPS KAKENHI Grant Number 16K16005.
  The third author was supported by JSPS KAKENHI Grant Number JP24106002, JP25280004, JP26280001, and JST CREST Grant Number JPMJCR1402.

\bibliography{ref}

\clearpage
\appendix

\section{Relationship Among \texorpdfstring{$m,\epsilon$ and $R$}{m, epsilon and R} in Algorithm~\ref{alg:weighted_general}}\label{sec:relation_m_eps_R}

Here, we prove some relationships among $m,\epsilon$ and $R$ in Algorithm~\ref{alg:weighted_general}.
\begin{lemma}\label{lem:relation_m_eps_R}
	Let $R\ge 3$, $m \coloneqq \lfloor (R-3)/2\rfloor$ and let $\epsilon > 0$ be a real such that $\log_{1+\epsilon} (1/\epsilon) = m$.
	Then, $m=\Theta(\frac{1}{\epsilon}\log\frac{1}{\epsilon})$ and $\epsilon = O(\log R/R)$
\end{lemma}

\begin{proof}
    By the definition of the base of natural logarithm $e$ and the monotonicity of $(1+1/x)^x$,
	we have $2\le (1+1/x)^x \le e$ for any $x\ge 1$.
	As $\epsilon\le 1$, we have
	\begin{align*}
		2^{\epsilon m}\le (1+\epsilon)^{\frac{1}{\epsilon}\epsilon m} &\le e^{\epsilon m}.
	\end{align*}
	By substituting $m = \log_{1+\epsilon} (1/\epsilon)$, we have $(1+\epsilon)^{\frac{1}{\epsilon}\epsilon m} = 1/\epsilon$.
	Hence, we get
	\begin{align}
		\epsilon m\log 2 \le \log \frac{1}{\epsilon} \le \epsilon m. \label{app:1}
	\end{align}
	This implies $m=\Theta(\frac{1}{\epsilon}\log\frac{1}{\epsilon})$.

	Next, we show that $\epsilon = O(\log R/R)$.
	By the inequalities \eqref{app:1}, we have
	\begin{align*}
		\epsilon &\le \frac{\log \frac{1}{\epsilon}}{m\log 2}
		 \le \frac{\log \bigl(\frac{1}{\epsilon}\log\frac{1}{\epsilon}\bigr)}{m\log 2}
		 \le \frac{\log m}{m\log 2}
		 = \frac{\log \lfloor (R-3)/2\rfloor }{\lfloor (R-3)/2\rfloor \log 2}
		 = O\left(\frac{\log R}{R}\right). \qedhere
	\end{align*}
\end{proof}

\end{document}